\documentclass[aps,pre,
               amsmath,amssymb,nofootinbib]{revtex4-2}
\usepackage[utf8]{inputenc}
\usepackage{natbib}
\usepackage{algorithm}
\usepackage{algorithmic}
\usepackage{amsmath}
\usepackage{amsthm}
\usepackage{amsfonts}
\usepackage{geometry}
\usepackage{verbatim}
\usepackage{tabularx}
\usepackage{graphicx}
\usepackage{hyperref}
\usepackage{cleveref}
\usepackage{todonotes}
\usepackage{siunitx}
\usepackage{tabularx,booktabs}

\newtheorem{theorem}{Theorem}
\newtheorem{lemma}[theorem]{Lemma}
\newtheorem{corollary}[theorem]{Corollary}

\newcommand{\expect}[1]{\left\langle #1 \right\rangle}
\newcommand{\bits}{\mathcal B}
\newcommand{\ints}{\mathbb Z}
\newcommand{\reals}{\mathbb R}
\newcommand{\enc}{e}
\newcommand{\mi}{:}
\newcommand{\dinexact}{{\mathchar'26\mkern-12mu d}}
\newcommand{\pf}[1]{\overline{#1}}
\newcommand{\lplus}{\stackrel{+}{<}}
\newcommand{\lmul}{\stackrel{\times}{<}}
\newcommand{\eqplus}{\stackrel{+}{=}}
\newcommand{\eqmul}{\stackrel{\times}{=}}
\newcommand{\gplus}{\stackrel{+}{>}}

\begin{document}

\title{Foundations of algorithmic thermodynamics}
\footnotetext[0]{An early draft of this article, titled “The algorithmic second law of thermodynamics”, was presented at the 17th International Conference on Computability, Complexity and Randomness in Nagoya, Japan.}
\author{Aram Ebtekar}
\email{aramebtech@gmail.com}
\affiliation{Vancouver, BC, Canada}

\author{Marcus Hutter}
\email{http://www.hutter1.net/}
\affiliation{Google DeepMind, London, UK}





\begin{abstract}
G{\'a}cs' coarse-grained algorithmic entropy leverages universal computation to quantify the information content of any given physical state. Unlike the Boltzmann and Gibbs-Shannon entropies, it requires no prior commitment to macrovariables or probabilistic ensembles, rendering it applicable to settings arbitrarily far from equilibrium. For measure-preserving dynamical systems equipped with a Markovian coarse-graining, we prove a number of fluctuation inequalities. These include algorithmic versions of Jarzynski's equality, Landauer's principle, and the second law of thermodynamics. In general, the algorithmic entropy determines a system's actual capacity to do work from an individual state, whereas the Gibbs-Shannon entropy only gives the mean capacity to do work from a state ensemble that is known a priori.
\end{abstract}


\maketitle

\section{Introduction}

Many of the most successful theories in physics have an \emph{initial value formulation}, in which the Universe is fully determined by its \emph{initial conditions} and \emph{dynamical equations of motion}. The second law of thermodynamics, despite being widely considered one of the most important facts of nature, does not appear explicitly in such formulations. Not only is it absent among the dynamical equations, but its irreversibility stands in contrast to the equations' charge-parity-time (CPT) symmetry \citep{lehnert2016cpt}. Nonetheless, if the second law is to hold for such formulations, then it must somehow follow from the initial condition and dynamics \citep{davies1977physics,albert2001time,ebtekar2024modeling}.

Another difficulty with the second law is its scope of applicability. Informally, it states that the \textbf{entropy} of an isolated physical system tends to increase. In order to apply this statement broadly, we require an unambiguous nonequilibrium definition of entropy. Focusing on classical (as opposed to quantum) mechanics for simplicity's sake, our definition should not depend on a prior choice of macrovariables or probabilistic ensemble (as the Boltzmann-Gibbs-Shannon entropies do). Moreover, the second law should not depend on properties such as nonequilibrium steady state or local detailed balance, that only hold in limited settings \citep{shiraishi2019fundamental,maes2021local,shiraishi2023introduction}.

Thus, we seek to define entropy as a function of the individual states in phase space, in such a way that a suitable initial value formulation makes it increase over time. Our task is made possible by two major insights from the scientific literature: one originating from the theory of dynamical systems, and the other from algorithmic information theory (AIT).

The first insight is the existence of \textbf{Markovian coarse-grainings}. These are memoryless partitions of a dynamical system's phase space into discrete cells. In more detail, we mean that the probability distribution of the system's coarse-grained state at any future time, conditional on its past and present, is given by evolving the dynamics forward from a uniform (i.e., proportional to the Liouville measure) distribution over the present cell. Whenever this property holds, we can take a discrete view of the system as a time-homogeneous Markov process. Taking the Liouville measure of each cell yields a discrete stationary measure for this process \citep{altaner2012microscopic}.

It remains an open problem to characterize when a coarse-graining is approximately Markovian \citep{nicolis1988master,werndl2009deterministic,figueroa2019almost,figueroa2021markovianization,strasberg2023classicality,ebtekar2024modeling}. To get some intuition, we can study toy systems whose coarse-graining is \emph{exactly} Markovian. For instance, \citet{altaner2012microscopic} introduce the \emph{network multibaker maps}: these are time-reversible deterministic chaotic dynamical systems, with microscopic randomization only in the initial state, whose coarse-graining emulates a wide variety of Markov chains. To illustrate how this occurs, we present a simplified version in \Cref{sec:apxmarkov}. These maps rigorously demonstrate that macroscopic irreversibility is compatible with microscopic reversibility, while also hinting at conditions under which more realistic systems might be approximately Markovian.

The Markov assumption is the basis for much of \textbf{stochastic thermodynamics} \citep{altaner2012microscopic,seifert2012stochastic,peliti2021stochastic,sagawa2022entropy,gaspard2022statistical,shiraishi2023introduction}, a powerful modern framework that replaces continuous-state dynamical systems by (usually discrete-state) Markov processes. These are easier to analyze: for Markov processes, the non-decrease of Gibbs-Shannon entropy is a straightforward and mathematically rigorous theorem, applicable to probability distributions arbitrarily far from equilibrium \citep[\S4.4]{thomas2006elements}.

However, while Markovian coarse-grainings motivate a probabilistic description of the \emph{dynamics}, nonequilibrium \emph{states} may lack an appropriate, non-subjective ensemble description, especially if they arise from an intricate computation; we elaborate on the reasons in \Cref{sec:probability}. To get an ensemble-free definition of entropy, as a function of individual physical states, we turn to a computability-based notion from the AIT literature. Intuitively speaking, the connection between physics and computability arises because the coarse-grained dynamics of our Universe are believed to have computational capabilities equivalent to a universal Turing machine \citep{deutsch2013computation,janzing2018does,kolchinsky2020thermodynamic,wolpert2024implications}.

\citet{gacs1994boltzmann} defines the coarse-grained \textbf{algorithmic entropy} of any individual state: roughly speaking, it is the number of bits of information that a fixed computer needs in order to identify the state's coarse-grained cell. For example, a state in which all particles are concentrated in one location would have low entropy, because the repeated coordinates can be printed by a short program. If the coarse-graining in question is Markovian, then \citet{levin1984randomness}'s law of \textbf{randomness conservation} says that the algorithmic entropy seldom decreases. In physical terms, we will come to see this as a vast generalization of the second law of thermodynamics.

In mathematical terms, it is an \emph{integral fluctuation} relation, meaning that it bounds an \emph{expectation} of the exponentiated entropy production. It is a statistical law that allows occasional small decreases in entropy. In stochastic thermodynamics, integral fluctuation relations are often derived as averages of corresponding \emph{detailed fluctuation} relations on individual state transitions \citep{seifert2012stochastic}. In \Cref{sec:apxrandom}, we state and prove the detailed fluctuation relations for randomness conservation, and show that they imply the integral relations. The remainder of this article explores the physical consequences of these relations, particularly when dealing with information-processing systems. Thus, we plant the seeds for a new computability-based view of thermodynamics, in which the entropy is a function of individual states rather than probability distributions.

\emph{Article outline.}  To start, \Cref{sec:background} places our work in the context of some relevant literature. \Cref{sec:preliminaries} reviews some notation, definitions, and useful facts.

\Cref{sec:theory} presents our theoretical contributions. Broadly speaking, we reformulate some aspects of stochastic thermodynamics in terms of G\'acs' coarse-grained algorithmic entropy. \Cref{sec:algoentropy} adapts G\'acs' definition to the Markovian setting, where \citet{levin1984randomness}'s randomness conservation applies. \Cref{sec:probability} compares the Gibbs-Shannon and algorithmic entropies, presenting conditions under which they coincide. When they do not coincide, we argue that a system's capacity to do work is in fact determined by the algorithmic entropy.

\Cref{sec:environment} introduces the concept of reservoirs that can exchange heat and work, and derives fluctuation inequalities for the algorithmic entropy flow and production. \Cref{sec:fluctuation} specializes these inequalities to the case of one reservoir at constant temperature. One of the inequalities generalizes \citet{kolchinsky2023}'s recent lower bound on the heat flow during a computation, while others can be seen as algorithmic versions of Jarzynski's equality \citep{jarzynski2011equalities} and Landauer's principle \citep{landauer1961irreversibility}, describing the exchange of heat, work, and information.

\Cref{sec:refine} extends randomness conservation to settings with variable dynamics or long time horizons, resulting in a fully general nonequilibrium second law of thermodynamics. As a special case, we recover a result of \citet{janzing2016algorithmic}, but with a different interpretation: we conclude that entropy is only produced when the coarse-grained dynamics are random. Our \Cref{cor:secondlaw} is arguably the most complete statement of the second law to date: it uses an ensemble-free notion of entropy, applies to arbitrary time intervals on a trajectory (but not its time-reversal), and allows fluctuations (which include Poincar\'e recurrence \citep{saussol2009introduction}).

\Cref{sec:discussion} presents some more applications. \Cref{sec:discussionintro} briefly discusses the effective dynamics of open systems, and introduces a few useful tricks for constructing examples. \Cref{sec:maxwell} applies our algorithmic second law to get an especially straightforward analysis of Maxwell's demon. \Cref{sec:landauer} discusses the thermodynamic costs of three kinds of information-processing: randomization, computation, and measurement. Finally, in order to demonstrate how a deficiency of algorithmic entropy serves as a resource, \Cref{sec:engine} models an information-theoretic analogue of a heat engine, which takes compressible strings as fuel.

\Cref{sec:conclusion} concludes with some possible directions for further research. The core mathematical ideas that we build upon, Markovian coarse-grainings and randomness conservation, are detailed in the appendices. While they are based on previous work \citep{gaspard1992diffusion,altaner2012microscopic,levin1984randomness,gacs2021lecture}, \Cref{sec:apxmarkov} considerably simplifies the network multibaker maps; meanwhile, \Cref{sec:apxrandom} states and proves randomness conservation in a manner that more closely parallels the thermodynamic fluctuation theorems \citep{seifert2012stochastic}. We hope that these ideas become more accessible as a result.

\section{Background}
\label{sec:background}

When Clausius first coined the term ``entropy'', thermodynamics was a macroscopic theory of work and heat transfer. The connection to information theory first arose in Szilard's response to Maxwell's famous thought experiment, in which a ``demon'' appears to violate the second law of thermodynamics by efficiently and intelligently tidying a system. Szilard's great insight was that any such tidying process must necessarily process information about the system. Once this information is correctly accounted for, the second law is restored \citep{maruyama2009colloquium}.

This accounting is usually done in an ad hoc manner. It remains difficult to define entropy in a unified manner that applies straightforwardly to Maxwell's demon and other nonequilibrium systems. The old definitions of Clausius and Boltzmann depend on specially chosen macrovariables, such as temperature, pressure, and chemical composition; these are well-suited to traditional equilibrium systems, but not to the demon's information storage.

Overly fine-grained microscopic definitions of entropy are equally unsuitable. Since the laws of physics are deterministic and time-reversible, fine-grained information can never be created or destroyed. Indeed, in classical mechanics, Liouville's theorem implies that the differential Gibbs-Shannon entropy is constant under Hamiltonian evolutions \citep{carcassi2020hamiltonian}. While quantum mechanics is beyond the scope of our article, we note that the von Neumann entropy is likewise constant under unitary evolutions. Although \citet[Appendix C]{zurek1989algorithmic} predicted a deterministic increase in the algorithmic entropy, which was later proved by \citet{janzing2016algorithmic}, \citet{gacs1994boltzmann} shows that this increase is negligible in practice; we will extend these results in \Cref{cor:janzing}.

A more suitable definition lies between the macroscopic and microscopic extremes, as one often finds a \emph{mesoscopic} coarse-graining that is approximately Markovian. A standard approach truncates the canonical positions and momenta, thus coarse-graining phase space into $6N$-dimensional hypercubes of measure $h^{3N}$, $h$ being Planck's constant and $N$ the number of particles \citep[\S12]{ehrenfest1959conceptual} \citep[\S2.7]{gaspard2022statistical}. In this coarse-grained view, the dynamics are effectively random, allowing the entropy to increase at up to the Kolmogorov-Sinai rate \citep{latora1999kolmogorov,entropybook,gaspard2022statistical}. \citet{nicolis1988master} and \citet{werndl2009deterministic} study some chaotic systems with Markovian coarse-grainings. In quantum mechanical settings, \citet{zurek1998decoherence} argues that decoherence has a similar effect.

At present, the Markov assumption appears to be necessary for much of thermodynamics. Indeed, due to CPT symmetry, entropy increase theorems that proceed directly from Hamiltonian dynamics tend to be too weak. \citet{gacs1994boltzmann} discusses a few of these. For example, ergodic systems are known to converge toward maximum entropy in both the infinite future and the infinite past. This fact does not distinguish the two temporal directions, and makes no comment on finite time intervals; in particular, it tells us nothing about how the entropy of yesterday should compare to the entropy of tomorrow.

He also discusses another kind of result, in which starting from a ``typical'' state guarantees that the entropy never falls below its initial value. \citet{kawai2007dissipation} describe a Hamiltonian formulation (see also \citet{esposito2010entropy} for a quantum mechanical version), in which the ``typical'' state is one where the environment is at equilibrium. For the Universe as a whole, \citet{albert2001time} envisions the Big Bang to meet the criteria for a typical state. The problem with this type of argument is that it only works once: subsequent states, having non-minimal entropy, are necessarily ``atypical''. From there, we have no reason to expect further increases in entropy.

Instead, we need the initial state to have the stronger property that its subsequent transitions forever retain ``typical'' statistics, long after the state itself becomes atypical. In other words, the coarse-grained trajectory should be a time-homogeneous Markov process. By modeling physical systems as Markov processes, we can formulate the second law over arbitrary time intervals.

From now on, we refer to coarse-grained \emph{states} or \emph{cells} interchangeably. In stochastic thermodynamics, the Gibbs-Shannon entropy is defined as a function, not of individual states, but of \emph{ensembles} that assign a probability $\mu(x)$ to each coarse-grained state $x$ \citep{seifert2012stochastic,peliti2021stochastic,sagawa2022entropy,gaspard2022statistical,shiraishi2023introduction}. Defining the Shannon codelength or \textbf{stochastic entropy} of each state $x$ by
\begin{equation*}
\hat H(x,\,\mu) := \log\frac{1}{\mu(x)},
\end{equation*}
the \textbf{Gibbs-Shannon entropy} is its expectation
\begin{equation*}
H(\mu)
:= \expect{\hat H(X,\,\mu)}_{X\sim\mu}
= \sum_{x} \mu(x)\log\frac{1}{\mu(x)}.
\end{equation*}

Note that these entropies are undefined for individual states $x$, unless a distribution $\mu$ is specified. Nonequilibrium choices of $\mu$ are typically arrived at by evolving a Markov process starting from a probabilistically prepared initial state. Such approaches have led to a number of major developments in the thermodynamics of information and computation \citep{sagawa2012thermodynamics,parrondo2015thermodynamics,wolpert2019stochastic,baumeler2019free,kolchinsky2020thermodynamic,parrondo2023information,manzano2024thermodynamics}. Nonetheless, the physical meaning of $\mu$ is not always clear \citep{frigg2019statistical}.

Entropy's main role in physics is to quantify a system's capacity to do work. Inspired by a thought experiment in which compressible data is used to do physical work, \citet{bennett1982thermodynamics} argues that a more precise and general definition of entropy should leverage universal computation to \emph{infer} the best description for any given state. Therefore, he defines the \textbf{algorithmic entropy} of an \emph{individual} state $x$ to be its Solomonoff-Kolmogorov-Chaitin description complexity $K(x)$, i.e., the length of the shortest program that outputs $x$ on a fixed universal computer. \citet{li2019introduction} study this function $K$ in detail.

\citet{zurek1989algorithmic} develops its physical interpretation and proves that
\begin{equation}
\label{eq:Hisanaverage}
H(\mu)\eqplus\expect{K(X\mid\mu)}_{X\sim\mu},
\end{equation}
where the equality holds up to an additive constant that does not depend on $\mu$. After obtaining a measurement result $x$ from a physical system, Zurek defines its entropy to be the sum $K(x)$ plus the posterior Gibbs-Shannon entropy given $x$. While this definition appears quite general, its dependence on measurements takes away from the objectivity of Bennett's definition.

\citet{gacs1994boltzmann} proposes a more natural refinement of Bennett's definition. First, he clarifies that the argument $x$ is a coarse-grained state. Second, when states occupy different Liouville measures $\pi(x)$ in the underlying phase space, he adds a correction term, defining the algorithmic entropy as
\begin{equation}
\label{eq:gacs}
S_\pi(x) := K(x) + \log\pi(x).
\end{equation}

We can view G\'acs' definition \labelcref{eq:gacs} as a special case of Zurek's, in which the measurement is a fixed function of the fine-grained state. If this function's range is countable, its elements correspond to cells of a coarse-graining. Committing to a fixed coarse-graining removes the need to actually perform measurements, making \labelcref{eq:gacs} purely a function of the individual coarse-grained state $x$. Thus, measurements are not built into the definition \labelcref{eq:gacs}, freeing us to consider arbitrary measurements as part of the dynamics, as we illustrate in \Cref{sec:maxwell}. We can also consider uncertain or random mesostates, as in stochastic thermodynamics: for random $X$, the entropy $S_\pi(X) = K(X) + \log\pi(X)$ becomes a random variable that depends on the value of $X$.

In settings where the Gibbs-Shannon and algorithmic entropies disagree, we must clarify their respective relationships to a system's capacity for work. In \Cref{sec:fluctuation}, we will see that the algorithmic entropy fundamentally determines the maximum amount of work that can be extracted from a specific physical state. Equation \labelcref{eq:Hisanaverage} then implies that the Gibbs-Shannon entropy measures the \emph{average} capacity for work, \emph{given} probabilistic knowledge of the state. In \Cref{sec:probability}, we elaborate on this distinction. Since the algorithmic entropy does not depend on averaging or priors, it carries a more objective physical meaning, which carries over to nonequilibrium settings where we lack probabilistic knowledge.

An important caveat is that the algorithmic entropy depends on a choice of universal computer. The dependence is bounded by the length of a compiler or interpreter between any pair of computers that we want to compare; fortunately, for realistic microprocessors, this length appears to be quite small. Indeed, entropy is measured in logarithmic units, such as bits, nats, or Joules per Kelvin \citep{frank2005indefinite}. The conversion rates are given by
\begin{equation}
\label{eq:conversion}
\SI{1}{\bit}
= k_B\ln 2
= \SI{9.57e-24}{\joule\per\kelvin},
\end{equation}
where $k_B$ is Boltzmann's constant, equivalent to 1 nat \footnote{The 2019 redefinition of SI units made $k_B := \SI{1.380649e-23}{\joule\per\kelvin}$ a ``defining constant'' \citep{newell2019international}. Effectively, the Kelvin became a derived unit, equal to \emph{exactly} $\SI{1.380649e-23}{}$ Joules per nat \citep{bedard2024temperature}.}. Pessimistically, consider a pair of computers, whose interpreters in each direction compress to about $\SI{12}{\gibi\byte}$ (i.e., $12\times 2^{33}$ bits). This is much larger than any practical interpreter known to the authors; and yet, even languages as distant as these would agree on the entropy of every system to within
\[12\times 2^{33}\times\SI{9.57e-24}{\joule\per\kelvin} < \SI{e-12}{\joule\per\kelvin}.\]

For macroscopic systems, this is a negligible difference. Of course, our notion of a ``realistic microprocessor'' should include physical size and resource constraints. Inspired by the Turing machine model, we imagine a microscopic control head operating on infinite-length tapes. Since its head is small, such a machine cannot cheat by ``hardcoding'' arbitrary data in its specification. Thus, the laws of nature may well determine which strings are considered simple or complex \citep{deutsch2013computation,janzing2018does}. We hope to make this argument rigorous in future work; see \citet[Appendix B]{zurek1989algorithmic} for a discussion of related issues.

Our article can also be compared with more recent works. \citet{baez2012algorithmic} apply thermodynamics to AIT, whereas we apply AIT to thermodynamics. \citet{kolchinsky2020thermodynamic} relate thermodynamic costs to description complexity; however, their approach depends on specific probabilistic priors, whereas we derive universal bounds from \citet{levin1984randomness}'s randomness conservation law. \citet{kolchinsky2023} proves an algorithmic detailed fluctuation inequality, which we show to follow from one of our own.

Throughout this article, we assume the existence of Markovian coarse-grainings and a suitable universal computer. The former allows us to substitute physical systems with their Markov process counterparts. Our modeling approach amounts to a minimalist version of stochastic thermodynamics, abstracting away many details of the physics to produce simple rigorous statements under minimal assumptions.

\section{Preliminaries}
\label{sec:preliminaries}

We start with some notation. $\ints$, $\mathbb Q$, and $\reals$ denote the integers, rational numbers, and real numbers, respectively. $\ints^+$, $\mathbb Q^+$, and $\reals^+$ denote their respective nonnegative subsets. $\ints_m:=\{0,1,\ldots,m-1\}$ denotes the first $m$ elements of $\ints^+$. Let $\bits:=\{\mathtt 0,\mathtt 1\}\simeq\ints_2$; its Kleene closure $\bits^*$ is the set of all finite-length binary strings. For a string $x\in\bits^*$, $|x|$ denotes its length in bits. For a set $A$, $|A|$ denotes its cardinality. Juxtaposition of strings $xy$ indicates their concatenation. When $f$ is a two-argument function, $f(\cdot,\,x)$ denotes the one-argument function that maps $y\mapsto f(y,\,x)$.

The capital letters $X,Y$ refer to random variables, while the lowercase variables $x,y$ refer to the specific values they take on. Expectations of random variables are denoted by angled brackets $\expect{\cdot}$. The probability of an event $E$, conditional on another event $F$, is denoted by $\Pr(E\mid F)$. The expression $\Pr(Y\mid X)$ is interpreted as a random variable, whose value is $\Pr(Y=y\mid X=x)$ whenever $Y=y$ and $X=x$. Similarly, the conditional expectation $\expect{\cdot\mid X}$ is a random variable that depends on the value of $X$.

\subsection{Stochastic matrices}
\label{sec:prelimmatrices}

Stochastic thermodynamics takes place on coarse-grained state spaces, which we represent as countable (i.e., finite or countably infinite) sets $\mathcal X,\mathcal Y$. Probabilities of transitions from $x\in\mathcal X$ to $y\in\mathcal Y$ are given by a \textbf{stochastic matrix} $P:\mathcal Y\times\mathcal X\rightarrow\reals^+$, satisfying
\[\forall x\in\mathcal X,\quad \sum_{y\in\mathcal X} P(y,\,x)=1.\]

Its action on a discrete \textbf{measure} $\pi:\mathcal X\rightarrow\reals^+$ takes it to a successor measure $P\pi:\mathcal Y\rightarrow\reals^+$, given by
\begin{equation}
\label{eq:dualpi}
P\pi(y) := \sum_{x\in\mathcal X} P(y,\,x)\pi(x).
\end{equation}
If $\sum_{x\in\mathcal X} \pi(x) = 1$, $\pi$ is called a \textbf{probability measure} (or distribution, mixture, or ensemble), and it follows that $P\pi$ is also a probability measure.

Given $\pi$ and $P$ such that $P\pi$ is nonzero everywhere, we define the \textbf{dual matrix} $\widetilde P:\mathcal X\times\mathcal Y\rightarrow\reals^+$ \citep{kolmogoroff1936theorie,nagasawa1964time} by
\begin{equation}
\label{eq:dualP}
\widetilde P(x,\,y) := \frac{P(y,\,x)\pi(x)}{P\pi(y)}.
\end{equation}
\Cref{eq:dualpi,eq:dualP} imply that $\widetilde P$ is stochastic, and satisfies $\widetilde P(P\pi)=\pi$.

Note that if $x$ is sampled according to $\pi$, and $y$ according to $P(\cdot,\,x)$, then by Bayes' rule, $\widetilde P(x,\,y)$ gives the reverse transition probability of $x$ given $y$. However, we will often be interested in nonequilibrium settings, where $\pi$ is fixed and different from the distribution of $x$. In that case, the reverse probabilities are sensitive to the distribution of $x$, and not equal to $\widetilde P$ in general \footnote{On the other hand, $P$ can be viewed as the coarse-grained evolution of an underlying dynamical system, as in \Cref{sec:apxmarkov}. Suppose that at a time $t_0$, the system's state is set to a continuous distribution over its phase space. There is evidence to suggest that the resulting coarse-grained trajectory will evolve forward by $P$ at times $t > t_0 + t_m$, and backward by $\widetilde P$ at times $t < t_0 - t_m$. The unpublished manuscript \citep{ebtekar2021information} proves this for certain extensions of the multibaker map, establishing a time-reversal symmetry centered near $t_0$. The ``microscopic mixing time'' $t_m$ depends on the initial distribution's smoothness, and can be far shorter than the time needed to reach macroscopic equilibrium (i.e., ``heat death''). We can think of $t_0$ as analogous to the Big Bang, and $\widetilde P$ as a coarse-graining of the conjectured CPT-inverted dynamics preceding it \citep{boyle2018cpt,boyle2022big,ebtekar2024modeling}.}.

For the remainder of this article (except in \Cref{sec:apxrandom}), we take $\mathcal X=\mathcal Y$. If in addition, $P\pi = \pi$, then we say $P$ is \textbf{$\pi$-stochastic}, or $\pi$ is \textbf{stationary} for $P$. In that case, $\widetilde P$ is also $\pi$-stochastic because $\widetilde P\pi = \widetilde P(P\pi)=\pi$. \textbf{Doubly stochastic} is a common synonym for $\sharp$-stochastic, where $\sharp$ denotes the \textbf{counting measure}, i.e., $\sharp(x) := 1$ for all $x$.

Finally, we say that $P$ satisfies \textbf{detailed balance} with respect to $\pi$ if
\begin{equation}
\label{eq:detailedbalance}
\forall x,y\in\mathcal X,\quad P(x,\,y)\pi(y) = P(y,\,x)\pi(x).
\end{equation}
This is a strong condition, equivalent to having both $P\pi=\pi$ and $P=\widetilde P$.

In physical settings, $\pi$ is the Liouville measure, and $\widetilde P$ is the coarse-grained dynamics after CPT inversion. Therefore, detailed balance occurs when all parity-odd microvariables (e.g., momenta) are coarse-grained away, leaving only bidirectional transitions \citep{shiraishi2019fundamental,shiraishi2023introduction}. In this article, $\pi$ is stationary by Liouville's theorem, but we allow detailed balance to fail; i.e., $P\pi=\pi$, but possibly $P\ne\widetilde P$.

\subsection{Markov processes}
\label{sec:prelimmarkov}

An $\mathcal X$-valued \textbf{stochastic process} is a collection of $\mathcal X$-valued random variables indexed by continuous time $(X_t)_{t\in\reals^+}$, or by discrete time $(X_t)_{t\in\ints^+}$. We say it is a \textbf{time-homogeneous Markov process} if
\begin{equation*}
\label{eq:transitionmatrix}
s\le t\implies \Pr(X_t \mid X_{\le s}) = \Pr(X_t \mid X_s) = P_{t-s}(X_t,\,X_s),
\end{equation*}
where the stochastic matrix $P_{\Delta t}:\mathcal X\times\mathcal X\rightarrow\reals^+$ is called the \textbf{transition matrix} for time steps of duration $\Delta t$. The first equality is called the \textbf{Markov property}, while the second expresses \textbf{time-homogeneity}. For any finite sequence of times $0=t_0<t_1<\cdots<t_n$, the chain rule of probability yields
\begin{equation}
\label{eq:chainrule}
\Pr(X_{t_0},\,\ldots,\,X_{t_n})
= \Pr(X_0)\prod_{i=1}^n P_{t_i - t_{i-1}}(X_{t_i},\,X_{t_{i-1}}).
\end{equation}

A discrete-time Markov process is also called a \textbf{Markov chain}. For $\Delta t\in\ints^+$, we have
\begin{equation}
\label{eq:iterate}
P_{\Delta t} = (P_1)^{\Delta t}.
\end{equation}
Therefore, \labelcref{eq:chainrule} implies that a Markov chain's joint probability distribution is uniquely determined by the distribution of $X_0$ and the matrix $P_1$; these are its \textbf{initial condition} and \textbf{dynamics}, respectively. A continuous-time Markov jump process is described similarly, with $P_{\Delta t}$ ($\Delta t\in\reals^+$) instead being generated by a transition \emph{rate} matrix \citep{shiraishi2023introduction}.

\subsection{Stochastic thermodynamics}
\label{sec:prelimstochastic}

We now present a minimalist version of the stochastic thermodynamics framework \citep{seifert2012stochastic,peliti2021stochastic,sagawa2022entropy,gaspard2022statistical,shiraishi2023introduction}. Classical physics studies continuous-time trajectories over the continuous phase space of a dynamical system. The system's evolution is deterministic and reversible, and preserves the phase space's \textbf{Liouville measure}. By coarse-graining the phase space into discrete cells, we can instead describe the trajectory as a stochastic process (in either discretized or continuous time), that jumps between a countable set of coarse-grained states $\mathcal X$. By integrating the Liouville measure over each cell $x\in\mathcal X$, we obtain a discrete measure $\pi:\mathcal X\rightarrow\reals^+$.

We assume the coarse-grained trajectory is a Markov process, whose transition probabilities $P_{\Delta t}(y,\,x)$ are equal to the fraction of the cell $x$, whose evolution after time $\Delta t$ is in the cell $y$. This assumption is an intense area of study \citep{nicolis1988master,werndl2009deterministic,altaner2012microscopic,figueroa2019almost,figueroa2021markovianization,strasberg2023classicality}, but in stochastic thermodynamics it is commonly taken for granted. It then follows that the Markov process is time-homogeneous and $\pi$-stochastic.

To keep things simple and general, we do not make any sort of steady state or detailed balance assumption. At this stage, we do not even define the concepts of energy or heat. In this generic setting, the \textbf{stochastic entropy} of a state $x$, sampled from a probability distribution $\mu$, relative to the (not necessarily normalizable) stationary measure $\pi$, is defined by
\begin{equation}
\label{eq:stochentropy}
\hat H_\pi(x,\,\mu) := \log\frac{\pi(x)}{\mu(x)}.
\end{equation}

To avoid arithmetic singularities, assume $\pi$ and $\mu$ are nonzero everywhere. We may omit the subscript $\pi$ in cases where it equals the counting measure $\sharp$. In practice (e.g., see \Cref{sec:environment}), the nonuniformity of $\pi$ comes from modeling environment interactions. Thus, we can think of the contributions $\log\frac{1}{\mu(x)}$ and $\log\pi(x)$ as the entropy of a base system and its environment, respectively. The generalized \textbf{Gibbs-Shannon entropy} of $\mu$ is its mean stochastic entropy (i.e., negated Kullback-Leibler divergence) relative to $\pi$:
\begin{equation}
\label{eq:shannon}
H_\pi(\mu) := \expect{\hat H_\pi(X,\,\mu)}_{X\sim\mu} = \sum_{x\in\mathcal X} \mu(x)\log\frac{\pi(x)}{\mu(x)}.
\end{equation}

The stochastic entropy satisfies the following integral fluctuation theorem.

\begin{theorem}
\label{thm:lawH}
Let $X,Y$ be $\mathcal X,\mathcal Y$-valued random variables, and $P(y,\,x) := \Pr(Y=y \mid X=x)$. If the measures $\pi,\mu:\mathcal X\rightarrow\reals^+$, $\nu:\mathcal Y\rightarrow\reals^+$, and $P\pi$ are nonzero everywhere, then,
\begin{equation*}
\expect{2^{\hat H_\pi(X,\,\mu) - \hat H_{P\pi}(Y,\,\nu)} \mid X}
= \frac{\widetilde P\nu(X)}{\mu(X)}.
\end{equation*}
\end{theorem}

\begin{proof}
By definition,
\begin{align*}
\expect{2^{\hat H_\pi(X,\,\mu) - \hat H_{P\pi}(Y,\,\nu)} \mid X}
&= \expect{\frac{\pi(X)\nu(Y)}{P\pi(Y)\mu(X)} \mid X}
\\&= \expect{\frac{\widetilde P(X,\,Y)\nu(Y)}{P(Y,\,X)\mu(X)} \mid X}
\\&= \sum_{y\in\mathcal Y}P(y,\,X)\frac{\widetilde P(X,\,y)\nu(y)}{P(y,\,X)\mu(X)}
\\&= \frac{\sum_{y\in\mathcal Y}\widetilde P(X,\,y)\nu(y)}{\mu(X)}
\\&= \frac{\widetilde P\nu(X)}{\mu(X)}.
\end{align*}
\end{proof}

In \Cref{thm:lawH}, $\mu$ and $\nu$ need not be related to the distribution of $(X,Y)$. \citet{roldan2023martingales} apply it to obtain martingales, essentially by setting $\mu:=\widetilde P\nu$. However, we choose $\mu$ to be the initial state distribution. Under the law of a $\pi$-stochastic matrix $P$, the distribution then evolves to $P\mu$. Therefore, the \textbf{stochastic entropy production} during a state transition $x\rightarrow y$ is defined by
\begin{align*}
\Delta\hat H_\pi
&:=\hat H_\pi(y,\,P\mu) - \hat H_\pi(x,\,\mu).
\end{align*}
By the law of total expectation, and \Cref{thm:lawH} with $P\pi=\pi$ and $\nu:=P\mu$,
\begin{equation}
\label{eq:stochIFT}
\expect{2^{-\Delta\hat H_\pi}}
= \expect{\expect{2^{-\Delta\hat H_\pi} \mid X}}
= \expect{\frac{\widetilde PP\mu(X)}{\mu(X)}}_{X\sim\mu}
= \sum_{x\in\mathcal X}\widetilde PP\mu(x)
= 1.
\end{equation}
By Jensen's inequality, it follows that the \textbf{Gibbs-Shannon entropy production} is nonnegative:
\[\Delta H_\pi:=H_\pi(P\mu)-H_\pi(\mu)=\expect{\Delta\hat H_\pi}\ge 0.\]

These results are quite general, applying to distributions that are arbitrarily far from equilibrium. Their main downside is that the stochastic \labelcref{eq:stochentropy} and Gibbs-Shannon \labelcref{eq:shannon} entropies depend on a choice of distribution $\mu$, and hence are not well-defined functions of the individual state $x$. While thermodynamic ensembles certainly have their place, \Cref{sec:probability} discusses settings in which an adequate choice of $\mu$ is not readily available. To get an ensemble-free notion of entropy, we now turn to algorithmic information theory.

\subsection{Algorithmic information theory}
\label{sec:prelimait}

A set of strings is \textbf{self-delimiting} if no element is a proper prefix of another. For an arbitrary set $A$, an \textbf{encoding} is an injective function $f:A\rightarrow\bits^*$; we say $f$ is self-delimiting if its range is.

For integers $n\in\ints^+$, a self-delimiting encoding $\pf n\in\bits^*$ is given recursively by
\begin{equation}
\label{eq:prefixfree}
\pf n := \begin{cases}
\mathtt 0 &\text{if }n=0,
\\\mathtt 1\pf{|B(n)|}B(n) &\text{if }n>0,
\end{cases}
\end{equation}
where $B(n)$ is the standard binary encoding of $n$ without its leading $\mathtt 1$. For example,
\[\pf 9
= \mathtt{1}\pf 3\mathtt{001}
= \mathtt{11}\pf 1\mathtt{1}\mathtt{001}
= \mathtt{111}\pf 0\mathtt{1}\mathtt{001}
= \mathtt{11101001}.\]

For strings $x\in\bits^*$, a self-delimiting encoding is given by $\pf x:=\pf{|x|}x$, equal to $\pf{B^{-1}(x)}$ without its leading $\mathtt 1$. For $A,B\subset\bits^*$, let $AB:=\{ab:A\in A,\,b\in B\}$. If $A$ is self-delimiting, the pair $(a,b)$ is uniquely decodable from $ab\in AB$. If both $A$ and $B$ are self-delimiting, then so is $AB$.

A \textbf{prefix machine} $T$ is a computer whose set of valid halting programs $\mathcal P_T\subset \bits^*$ is self-delimiting. Let $T(p)$ denote the output of $T$ on $p\in\mathcal P_T$, and write $T(p)=\emptyset$ for $p\in\bits^*\setminus\mathcal P_T$. We fix a \textbf{universal} prefix machine $U$ with the following property: for every prefix machine $T$, there exists $x_T\in\bits^*$, such that for all $y,p\in\bits^*$,
\[U(\pf y\,\pf{x_T}p) = T(\pf yp).\]
See Hutter \emph{et al.}~\citep{hutter2004universal,hutter2024introduction} for the details of this construction, or \citet{li2019introduction} for a similar approach. Since $\pf yp$ is uniquely decodable, from now on we instead write $(y,p)$.

The \textbf{description complexity} of a string $x\in\bits^*$, given side information $y\in\bits^*$, is
\[K(x\mid y) := \min_p\{|p|:\, U(y,p)=x\}.\]

When $y$ is the empty string, $K(x\mid y)$ becomes the unconditional description complexity $K(x)$. Whenever an encoding is implied, we may write non-string objects in place of $x$ or $y$. For example, a finite set is encoded by a lexicographic listing of its elements. A computable function (or measure or matrix) is encoded by any program that computes it in the sense of \labelcref{eq:computable} below. While the resulting complexity depends on which program is chosen, our derivations remain valid provided that repeated mentions of a function always refer to the same program. We also assume fixed encodings for the countable sets $\ints$, $\mathbb Q$, and $\mathcal X$, so that their elements have well-defined description complexities. Complexities and entropies in this article are measured in units of bits; accordingly, all logarithms implicitly have base $2$.

(In)equalities that hold up to constant additive terms or multiplicative factors are expressed by writing a $+$ or $\times$ on top of the (in)equality sign. For example, $f(x) \lplus g(x)$ and $f(x) \lmul g(x)$ mean $f(x) < c+g(x)$ and $f(x) < c\cdot g(x)$, respectively, for some constant $c$. By ``constant'', we mean that $c$ is a function of only the parameters that we explicitly declared as fixed, such as the universal computer $U$ and the encodings of important sets. We say $x$ is \textbf{simply describable} from an optional context $y$, if $K(x\mid y)\eqplus 0$.

We review some properties of $K$. Since programs are self-delimiting, we have Kraft's inequality
\begin{equation}
\label{eq:kraft}
\sum_{x\in\bits^*} 2^{-K(x\mid y)} < 1.
\end{equation}

There exists a decoder $p$ for \labelcref{eq:prefixfree}, such that for all $n\in\ints^+$, $U(p\pf n) = n$. Hence,
\begin{equation*}
K(n)
\le |p\pf n|
\eqplus |\pf n|
= 1 + \sum_{i=1}^\infty\left\lfloor1 + \log^i n\right\rfloor \lplus \log n + 2\log\log n,
\end{equation*}
where the sum is evaluated as far as the $i$-fold iterated logarithm is nonnegative, and the rightmost inequality assumes $n\ge 2$. Similarly, for $x\in\bits^*$, $K(x)\lplus|\pf x|\lplus |x| + 2\log|x|$.

Whereas the Gibbs-Shannon entropy measures the mean information content per independent sample of a probability distribution \citep{thomas2006elements}, the description complexity measures the information content of an individual string without reference to any distribution. Naturally, it satisfies analogous relations \citep{grunwald2004shannon}. We make frequent use of the following:
\begin{align*}
K(x\mid y) \lplus K(x)
\lplus K(x,y) \eqplus K(y,x) \eqplus K(y) + K(x\mid y,K(y)) \lplus K(y) + K(x\mid y).
\end{align*}

The algorithmic \textbf{mutual information} between $x$ and $y$ is defined by
\begin{equation}
\label{eq:mutinf}
I(x\mi y) := K(x)+K(y)-K(x,y) \eqplus K(x)-K(x\mid y,K(y)).
\end{equation}
The conditional mutual information $I(x\mi y\mid z)$ is defined by conditioning on $z$ every $K$ term in \labelcref{eq:mutinf}. It can be shown to satisfy a data processing identity \citep{gacs2021lecture}:
\begin{equation}
\label{eq:dataproc}
I(x\mi (y,z)) \eqplus I(x\mi z) + I(x\mi y\mid z,K(z)).
\end{equation}

Finally, for any computable probability measure $\mu$, and $\delta > 0$,
\begin{equation*}
\label{eq:universal}
\delta\cdot 2^{-K(x\mid \mu)}
< \mu(x)
\lmul 2^{-K(x\mid \mu)}
,
\end{equation*}
where the first inequality fails on a set of $\mu$-probability less than $\delta$ \footnote{This follows by summing over those $x$ which violate the inequality, and applying Kraft's inequality \labelcref{eq:kraft}.}, while the second holds for all $x\in\bits^*$ \citep[\S4.3]{li2019introduction}. Hence, with $\mu$-probability greater than $1-\delta$,
\begin{equation}
\label{eq:loguniversal}
K(x\mid \mu)
\lplus \log\frac{1}{\mu(x)}
< K(x\mid \mu) + \log\frac{1}{\delta}.
\end{equation}
In particular, setting $\mu$ to a uniform distribution on any finite set $A\subset\bits^*$ reveals that $K(x\mid A) \lplus \log |A|$, with $K(x\mid A) \eqplus \log |A|$ for all but a constant fraction of $x\in A$. For ``simple'' $\mu$ satisfying $K(\mu)\approx 0$, we have $K(x\mid\mu)\approx K(x)$; hence, \labelcref{eq:loguniversal} lets us interpret $K(x)$ as a universal variant of the Shannon codelength, that does not depend on a choice of $\mu$.

Note that we stated all these relations with respect to an arbitrary universal computer $U$. By applying them to the universal computer $U_z(y,p) := U((z,y),p)$, we see that they remain valid when every mention of $K$ is conditioned on any additional data $z\in\bits^*$.

Short programs that output $x$ serve as compressed representations of $x$. Consider the \textbf{universal compression algorithm}: for a fixed time budget, simulate all programs of length up to about $|x| + 2\log|x|$, and then output the shortest program that halted with output $x$. As the time budget increases to infinity, the resulting program length decreases to $K(x)$. Thus, we may think of $K(x)$ as the optimum lossless compression achievable for $x$, in the limit of infinite runtime. In the spirit of Occam's razor, suppose we think of the shortest program as ``explaining'' $x$ \citep{rathmanner2011philosophical}; then, we see that explanations are falsified in finite time, but never proven, as we can never rule out the possibility that $x$'s shortest program is among those that have yet to halt.

A function $f:\bits^*\rightarrow\bits^*$ (or between countable sets associated with encodings) is \textbf{computable} if there exists a prefix machine $T$, such that $T(x) = f(x)$ for all $x\in\bits^*$. The description complexity $K$ is not computable in this sense; however, our universal compression algorithm is easily adapted to compute a decreasing integer sequence that approaches $K(x)$ from above. Hence, we say that $K$ is \textbf{upper semicomputable}.

To extend these concepts to real-valued functions, we say $f:\bits^*\rightarrow\reals$ is \textbf{lower (upper) semicomputable} if there exists a computable function $g:\bits^*\times\ints^+\rightarrow\mathbb Q$, such that $g(x,\,\cdot)$ is monotonically increasing (decreasing), and
\[\lim_{n\rightarrow\infty} g(x,\,n) = f(x).\]
The real-valued function $f$ is \textbf{computable} if it is both lower and upper semicomputable; or equivalently, if there exists a computable function $g$ such that, for any desired level of precision $n$,
\begin{equation}
\label{eq:computable}
\left|g(x,\,n) - f(x)\right| < 2^{-n}.
\end{equation}

\section{Theoretical analysis}
\label{sec:theory}

\subsection{Coarse-grained algorithmic entropy}
\label{sec:algoentropy}

As in \Cref{sec:prelimstochastic}, we model a physical system's coarse-grained trajectory by a Markov process on the state space $\mathcal X$. We also fix a canonical encoding by which to identify $\mathcal X$ with a subset of $\bits^*$. It will be convenient to define a conditional version of \citet{gacs1994boltzmann}' coarse-grained algorithmic entropy \labelcref{eq:gacs}. Therefore, let the \textbf{algorithmic entropy} of $x\in\mathcal X$, given side information $z\in\bits^*$, relative to the stationary measure $\pi:\mathcal X\rightarrow\reals^+$, be
\begin{equation}
\label{eq:entropy}
S_\pi(x\mid z) := K(x\mid z) + \log \pi(x).
\end{equation}
Note that \labelcref{eq:entropy} is formally identical to the stochastic entropy \labelcref{eq:stochentropy}, if we replace the prior $\mu(x)$ with $2^{-K(x\mid z)}$.

The equilibrium properties of $S_\pi$ are fairly straightforward. Assuming $Z:=\sum_{x\in\mathcal X}\pi(x)<\infty$, the normalization $\pi(x)/Z$ is an equilibrium ensemble. Further assuming that $\pi(x)/Z$ is simply describable from $z$, substituting it into \labelcref{eq:loguniversal} yields
\begin{equation}
\label{eq:Kequib}
K(x\mid z)
\lplus \log\frac{Z}{\pi(x)}
< K(x\mid z) + \log\frac 1\delta,
\end{equation}
where the first inequality holds for every state $x\in\mathcal X$, while the second fails with equilibrium probability less than $\delta$. Define the $\pi$-randomness deficiency \footnote{This concept originates in the algorithmic randomness literature \citep{levin1976uniform,vereshchagin2004kolmogorov,gacs2005uniform}. In the language of statistical hypothesis testing, randomness deficiency is the logarithm of an \emph{e-value} or \emph{likelihood ratio} between a null hypothesis $\pi(x)/Z$, and a ``universal'' alternative hypothesis $2^{-K(x\mid z)}$  \citep{vovk2020non,ramdas2023game,ramdas2024hypothesis}. Thus, $J_\pi(x\mid z)$ quantifies how implausible it is that $x$ was sampled from the equilibrium distribution, if $z$ was known prior to sampling.}, or \textbf{negentropy}, by
\begin{equation}
\label{eq:negentropy}
J_\pi(x\mid z)
:= \log Z - S_\pi(x\mid z)
= \log\frac{Z}{\pi(x)} - K(x\mid z).
\end{equation}
Then, \labelcref{eq:Kequib} becomes simply
\begin{equation}
\label{eq:negentropysaturates}
0 \lplus J_\pi(x\mid z) < \log\frac 1\delta.
\end{equation}
Therefore, the algorithmic entropy $S_\pi$ has a maximum value of about $\log Z$, and at equilibrium it concentrates near this maximum. Under any of the standard ergodic conditions that make a Markov process converge to $\pi$, it follows that the algorithmic entropy tends toward this maximum, and then rarely fluctuates away from it.

In the general nonequilibrium setting, the randomness conservation theorem of \citet{levin1984randomness} says that the entropy tends not to decrease. For our purposes, a particular formulation is helpful, which we present as \Cref{thm:lawS} in \Cref{sec:apxrandom}. There, we consider the transition matrix $P$ that transforms a random earlier state $X$ of the process to a later state $Y$. If $P$ is $\pi$-stochastic, with both $\pi$ and $P$ being computable, then \Cref{thm:lawS} says that
\begin{equation}
\label{eq:flucS}
\expect{2^{S_\pi(X\mid\widetilde P)-S_\pi(Y\mid\widetilde P)}} \lmul 1.
\end{equation}
Moreover, for all $\delta > 0$, with probability greater than $1-\delta$,
\begin{equation}
\label{eq:lawS}
S_\pi(X\mid\widetilde P)-S_\pi(Y\mid\widetilde P) \lplus \log\frac 1\delta.
\end{equation}

Two remarks are in order. Firstly, unlike its stochastic analogue \labelcref{eq:stochIFT}, the algorithmic fluctuation relation \labelcref{eq:flucS} need not hold with equality. For example, suppose $P(y,\,x):=1/|\mathcal X|$ for all $x,y$ in some large finite set $\mathcal X$. If the earlier state has $K(X\mid\widetilde P)\eqplus 0$ with probability one, then
\begin{equation*}
\label{eq:randomizerexample}
\expect{2^{S_\sharp(X\mid\widetilde P)-S_\sharp(Y\mid\widetilde P)}}
\eqmul \expect{2^{-K(Y\mid\widetilde P)}}
= \frac{1}{|\mathcal X|}\sum_{y\in\mathcal X}2^{-K(y\mid\widetilde P)}
< \frac{1}{|\mathcal X|}
\ll 1.
\end{equation*}

Secondly, \labelcref{eq:lawS} bounds the increase in entropy by a constant plus $\log(1/\delta)$. The constant term comes from basic properties of $K$, and can be made very small by an appropriate choice of reference computer $U$ \citep[\S3.9]{li2019introduction}. The $\log(1/\delta)$ term is also quite small: supposing we tolerate $\delta = 2^{-1000}$, it amounts to a kilobit, which is negligible in terms of physical units (recall \labelcref{eq:conversion}). Therefore, at macroscopic scales, we can effectively say that entropy never decreases.

The conditional parameter, consisting of a program for $\widetilde P$, remains somewhat of a nuisance. We want a fixed entropy function with which to compare states encountered throughout our system's evolution. Conditioning on $\widetilde P$ would be fine if it were solely determined by the laws of physics; but unfortunately, it also depends on the time elapsed between the two snapshots $X,Y$ of our system. In \Cref{sec:refine}, we derive the appropriate correction for this dependence.

In the meantime, we deal with the more urgent matter of choosing a suitable coarse-graining, and equipping it with a string encoding. We consider two types of coarse-graining: \textbf{macroscopic} and \textbf{mesoscopic}. \citet{gacs1994boltzmann} focuses primarily on the macroscopic type, letting the elements of $\mathcal X$ be the Boltzmann macrostates. In other words, given a physical system, we imagine measuring a few specially designated macrovariables, such as the temperature, pressure, and chemical composition of each of its parts, to a reasonable number of significant figures. The phase space is then partitioned into cells corresponding to every possible joint measurement outcome.

Since each cell $x\in\mathcal X$ is determined by the values of its macrovariables, any standard numerical encoding (e.g., binary scientific notation) would work. $K(x\mid\widetilde P)$ then is quite small: even writing say a hundred macrovariables, each to a dozen significant figures, requires only a few kilobits. Therefore, the algorithmic entropy \labelcref{eq:entropy} of a \emph{macrostate} $x$ simplifies to
\[S_\sharp(x\mid\widetilde P) \approx \log\pi(x),\]
which is just its Boltzmann entropy.

For physical systems whose Boltzmann macrostates mix sufficiently rapidly, we expect this coarse-graining to be approximately Markovian. However, there are systems whose Boltzmann macrostates are not ergodic. For example, computer systems are heavily dependent on the stability of their memory states. As a result, every memory state must be treated separately, even if they occur at the same ambient temperature, pressure, and so on. Here, we might use the coarse-graining $\mathcal X:=\bits^m$, corresponding to all possible settings of an $m$-bit memory.

In more general settings, we can obtain a mesoscopic coarse-graining by rounding or truncating the values of the canonical microvariables \citep[\S12]{ehrenfest1959conceptual} \citep[\S2.7]{gaspard2022statistical}. For an $N$-particle Hamiltonian system, each cell $x\in\mathcal X$ is determined by the individual particles' $3$-dimensional positions and momenta, to a high but finite level of precision. In phase space, these cells are tiny $6N$-dimensional hypercubes. The string encoding consists of $6N$ numerical positions and momenta, written in any standard format. The cells have equal Liouville measure, making $\pi$ constant; by normalizing, we can set $\pi:=\sharp$, so that $\log\sharp(x)=0$ and $\widetilde P(x,\,y) = P(y,\,x)$. Therefore, the algorithmic entropy \labelcref{eq:entropy} of a \emph{mesostate} $x$ simplifies to
\[S_\sharp(x\mid\widetilde P) = K(x\mid\widetilde P)\eqplus K(x\mid P),\]
which is just the description complexity of its canonical microvariables.

We apply the macroscopic coarse-graining to \textbf{reservoir} systems whose macrostates mix rapidly, and the mesoscopic coarse-graining to all other systems. We can think of the mesostates as a refinement of the macrostates: each Boltzmann macrostate corresponds to a large set $B\subset\mathcal X$ of mesostates, with Boltzmann entropy $\log|B|$. Since $K(B)$ is small, \labelcref{eq:loguniversal} implies agreement between the algorithmic and Boltzmann entropies: for the vast majority of mesostates $x\in B$,
\begin{equation}
\label{eq:boltzmann}
K(x\mid P)\approx K(x\mid P,B)\eqplus\log |B|.
\end{equation}

One advantage of the mesoscopic description is that the macrovariables need not be chosen in advance. \labelcref{eq:boltzmann} holds not only for the classical Boltzmann macrostates, but for every simply describable finite set $B\subset\mathcal X$ containing $x$. For example, the entropy of a bookshelf may be estimated by taking $B$ to be the set of configurations compatible with how the books are sorted. For all $x\in B\subset\mathcal X$,
\begin{equation}
\label{eq:bestset}
K(x\mid P) \lplus K(B\mid P) + \log|B|,
\end{equation}
since $x$ is uniquely identified by a description of $B$, along with a numerical index of size $\log|B|$. In particular, the Boltzmann entropy is only an upper bound on $K(x\mid P)$. The latter may be smaller if the state $x$ has additional structure not captured by the Boltzmann macrovariables.

Generalizing the fixed-size index to a variable-size Shannon code, we also have that for all computable probability measures $\mu:\mathcal X\rightarrow\reals^+$,
\begin{equation}
\label{eq:bestmeasure}
K(x\mid P) \lplus K(\mu\mid P) + \log\frac{1}{\mu(x)}.
\end{equation}
In algorithmic statistics \citep[\S5.5]{li2019introduction} \citep{gacs2001algorithmic,wallace2005statistical,vereshchagin2016algorithmic}, the right-hand sides of \labelcref{eq:bestset,eq:bestmeasure} are minimized in order to \emph{infer} which set $B$ or ensemble $\mu$ best describes $x$. There even exist so-called \textbf{nonstochastic} strings $x$, for which no simply describable $B$ or $\mu$ makes the inequalities tight \citep{gacs2001algorithmic}. Thus, the algorithmic entropy takes into account much more general descriptions of states, than do the traditional Boltzmann and Gibbs-Shannon entropies.

\subsection{Comparison against Gibbs-Shannon entropy}
\label{sec:probability}

One reason why entropy is often regarded as mysterious is that its definitions present conceptual challenges. The Gibbs-Shannon entropy is a function of ensembles, and therefore makes no comment on individual physical states. The algorithmic entropy, on the other hand, is defined for all individual states, but fails to be computable. While every program that outputs $x$ yields an \emph{upper} bound on $K(x)$, there is \emph{no} $x$ for which we can compute a large \emph{lower} bound. For if there were, then a small algorithm could search for such $x$ and output the first one it finds, in contradiction to $K(x)$ being large. This is Chaitin's incompleteness theorem \citep[\S4]{chaitin1974information} \citep[Theorem 1.5.2]{gacs2021lecture}.

While it appears to be a serious shortcoming, the incomputability of $K$ was shown to \emph{benefit} its application to universal prediction \citep[\S3]{solomonoff2009algorithmic}. To see how incomputability likewise benefits thermodynamics, suppose we instead define the entropy as some state function $K'$, for which large lower bounds \emph{are} computable. Then once again, a short program $p$ can search for and output an $x$ for which $K'(x)$ is large. Identifying $p$ with its smallest physical implementation, we expect $K'$ to satisfy the ``physical continuity'' property $K'(p,\,y)\approx K'(y)$ for all $y$. Any program that computes $x$ can also be used to erase a preexisting copy of $x$ \footnote{Following \citet{bennett1982thermodynamics}, the sequence of computations is $(p,x,0,0)\rightarrow(p,x,x,g)\rightarrow(p,0,x,g)\rightarrow(p,0,0,0)$. First, the program is run to produce a second copy of $x$ along with a computation trace $g$; then, the copy of $x$ is used to reversibly erase the original; finally, the program is run backward to clean up its outputs.}. As a result, $K'$ would decrease substantially, from $K'(p,\,x)\approx K'(x)$ to $K'(p,\,0)\approx 0$. Hence, such a $K'$ cannot have a second law.

Having seen why we need $K$, we are left with the practical matter of estimating it. While lower bounds cannot be computed for any \emph{individual} $x$, \labelcref{eq:loguniversal} provides a lower bound that holds for \emph{most} samples from any given distribution $\mu$. Moreover, by adding $\expect{\log\pi(X)}$ to both sides of \labelcref{eq:Hisanaverage} and presenting it alongside the definition \labelcref{eq:shannon}, we obtain
\begin{equation*}
H_\pi(\mu)
:= \expect{\hat H_\pi(X,\,\mu)}_{X\sim\mu}
\eqplus\expect{S_\pi(X\mid\mu)}_{X\sim\mu}.
\end{equation*}

Ignoring the middle expression for a moment, this says the Gibbs-Shannon entropy is an \emph{average over $\mu$}, of the algorithmic entropy \emph{conditioned on prior knowledge of $\mu$}. The two stipulations of averaging and prior knowledge cast doubt on whether $H_\pi$ captures anything objective about physical states. Fortunately, in a wide range of practical settings, there is a natural choice of ensemble $\mu:\mathcal X\rightarrow\reals^+$, corresponding to the individual underlying state $x\in\mathcal X$, such that a more direct equivalence holds:
\begin{equation}
\label{eq:HequalK}
H_\pi(\mu)
\approx\hat H_\pi(x,\,\mu)
\approx S_\pi(x).
\end{equation}

The idea is as follows. We consider $\mu$ to be a useful summary of $x$ when three conditions hold:
\begin{enumerate}
\item $\mu$ is simply describable, i.e., computable with $K(\mu)\approx 0$.
\item $\mu$ has a notion of \emph{typicality}: most of its samples share some characteristics of interest.
\item $x$ is among the typical samples with those characteristics.
\end{enumerate}

Condition 1 already implies
\begin{equation}
\label{eq:HtoK1}
S_\pi(x\mid\mu)\approx S_\pi(x),
\end{equation}
removing the need to condition on prior knowledge of $\mu$. It remains only to get rid of the averaging.

Adding $\log\pi(x)$ on all sides of \labelcref{eq:loguniversal} implies that
\begin{equation}
\label{eq:HtoK2a}
\Pr_{X\sim\mu}\left(S_\pi(X\mid\mu) \approx \hat H_\pi(X,\,\mu)\right)\approx 1,
\end{equation}
where we've hidden the tolerances behind approximation ($\approx$) signs for brevity. We take condition 2 (typicality) to mean that the stochastic entropy concentrates near its expectation:
\begin{equation}
\label{eq:HtoK2b}
\Pr_{X\sim\mu}\left(\hat H_\pi(X,\,\mu)\approx H_\pi(\mu)\right)\approx 1.
\end{equation}
In many practical settings, \labelcref{eq:HtoK2b} is derived as a consequence of the law of large numbers \citep[\S3]{thomas2006elements}.

Condition 3 is that $x$ is typical for $\mu$, which we take to mean that $x$ belongs to the intersection of the high-probability sets given by \labelcref{eq:HtoK2a,eq:HtoK2b}. That is,
\begin{equation}
\label{eq:HtoK3}
H_\pi(\mu)
\approx \hat H_\pi(x,\,\mu)
\approx S_\pi(x\mid\mu).
\end{equation}
When all three conditions hold, \labelcref{eq:HtoK1,eq:HtoK3} together imply \labelcref{eq:HequalK}; that is, the Gibbs-Shannon, stochastic, and algorithmic entropies all approximately coincide!

As an example, let $\mu$ be the canonical ensemble for a mechanically isolated container of an ideal gas at thermodynamic equilibrium. This ensemble's simple description meets condition 1, while the gas particles' independence and large number ensure condition 2. The overwhelming majority of states encountered at equilibrium are typical in the sense of condition 3. For those states, we conclude that the equivalence \labelcref{eq:HequalK} holds. The same argument applies to nonequilibrium ensembles, provided that they are simply describable and satisfy the concentration property \labelcref{eq:HtoK2b}.

On the other hand, we now consider three examples that break each respective condition. We use them to argue that, when the equivalence \labelcref{eq:HequalK} does not hold, $S_\pi(x)$ is a more physically correct measure of entropy than $H_\pi(\mu)$. For simplicity, let $\pi:=\sharp$, so that $S_\pi = K$.

The first violation, where $K(\mu)\gg 0$, is exemplified by letting $\mu$ be the point mass on a high-complexity state $x$. In this case, $K(x)\eqplus I(x\mi \mu)\eqplus K(\mu)\gg H(\mu)=0$. Intuitively, the Gibbs-Shannon approach takes $\mu$ as an exogenous parameter to specify that we \emph{know} the value of $x$, and can therefore clear it reversibly. In reality, in order to use any sort of knowledge, we must have it physically encoded in a memory device such as our brain. A strength of the algorithmic approach is that it naturally models knowledge as an endogenous part of the physical system. In \Cref{sec:maxwell}, we see how to model a Maxwell's demon that acquires, and then uses, information about $x$.

The second violation, where $\mu$ lacks a typical set, is exemplified by a robot that flips a hidden coin to decide whether or not to drain its battery. The probabilistic state $\mu$ of the battery and surrounding heat reservoir becomes an equal mixture of the two ensembles corresponding to a full or empty battery. Hence, $H(\mu)$ reaches an intermediate level: a free energy calculation in terms of Gibbs-Shannon entropy would suggest that the robot can do about half a charge worth of work. In reality, the robot will do either zero or a full charge of work. $H(\mu)$ merely predicts the \emph{average} work output among states sampled from $\mu$. In contrast, we see in \Cref{sec:fluctuation} that $K(x)$ predicts the work output from a \emph{specific} state $x$, whether it be zero or a full charge.

The final violation occurs when $x$ is atypical for $\mu$, in such a way that $K(x)\ll\hat H(x,\,\mu)$. Whether this is due to bad inductive priors or because $x$ is nonstochastic \citep{gacs2001algorithmic}, the Shannon code for $\mu$ is then suboptimal for $x$. As a result, a Gibbs-Shannon free energy calculation would underestimate the work that a suitable compression algorithm can extract from $x$. If we take this calculation too seriously, the compression algorithm would appear to violate the second law of thermodynamics.

It is interesting to observe that the so-called \textbf{universal measure} $\mathbf{m}(x):=2^{-K(x)}$ has a ``Shannon codelength'' of precisely $\hat H(x,\,\mathbf m) = K(x)$. This makes it useful as an inductive prior \citep{ebtekar2021information,hutter2004universal,rathmanner2011philosophical,wolpert2002supervised,muller2020law,hutter2024introduction}. However, regardless of whether we normalize $\mathbf{m}$, its Gibbs-Shannon entropy is physically meaningless because it lacks the concentration property \labelcref{eq:HtoK2b} \footnote{Indeed, if $\mathcal X$ is simply describable (as a subset of $\bits^*$) and large, then $H(\mathbf m)$ is also large. The program describing $\mathcal X$ enumerates its elements, of which the first $x\in\mathcal X$ satisfies $\hat H(x,\,\mathbf m) = \log\frac{1}{\mathbf m(x)} = K(x) \eqplus K(\mathcal X)\eqplus 0$. Thus, $x$ has substantial ``probability'' $\mathbf m(x)$, despite having $\hat H(x,\,\mathbf m)\ll H(\mathbf m)$.}. For additional comparisons between the probabilistic \citep{thomas2006elements} and algorithmic \citep{li2019introduction} flavors of information theory, see \citep{grunwald2004shannon,gacs2001algorithmic,rathmanner2011philosophical,austern2019gaussianity}. We adopt the view of \citet{kolmogorov1983combinatorial}, who considered algorithmic descriptions to be conceptually prior to (and more general than) probabilistic ones. In his words:

\begin{quote}
``Information theory must precede probability theory, and not be based on it. By the very essence of this discipline, the foundations of information theory have a finite combinatorial character.''
\end{quote}

\subsection{Interactions with general reservoirs}
\label{sec:environment}

Now, we specialize our framework to the commonly studied setting in thermodynamics, in which a \textbf{base system} interacts with an environment composed of one or more \textbf{reservoir systems}. Suppose the $i$'th reservoir's macrostate is determined by its \textbf{energy} $E_i\in\reals$, and possibly some additional macrovariables (e.g., volume and particle count) that we collectively denote by $\mathbf V_i\in\reals^{d_i}$ ($d_i\in\ints^+$). Let $\pi_i(E_i,\mathbf V_i)$ denote the Liouville measure of this macrostate, so that
\[B_i(E_i,\mathbf V_i)
:= k_B\ln\pi_i(E_i,\mathbf V_i)\]
is its Boltzmann entropy in physical units \footnote{For an infinite reservoir, the energy, volume, and Boltzmann entropy would be infinite. Fortunately, we only care about relative changes in these quantities, so we can normalize them to be finite.}. Our coarse-graining formalism restricts the macrovariables to a discrete set of values; nonetheless, if $B_i$ is approximately linear over typical increments in $(E_i,\mathbf V_i)$, then we can model it as a differentiable function.

Define the \textbf{temperature} $T_i$ by
\begin{equation}
\label{eq:temperature}
\frac{1}{T_i} := \left(\frac{\partial B_i}{\partial E_i}\right)_{\mathbf V_i}\!\!.
\end{equation}
At any state $(E_i,\,\mathbf V_i)$ for which $T_i\ne 0$, the implicit function theorem lets us locally write the energy as a differentiable function $E_i(B_i,\,\mathbf V_i)$, with
\begin{equation}
\label{eq:partials}
dE_i = T_i\, dB_i + \left(\frac{\partial E_i}{\partial\mathbf V_i}\right)_{B_i}\!\!\boldsymbol\cdot d\mathbf V_i.
\end{equation}

Thus, the flow of energy into the reservoir is a sum of two contributions. The \textbf{heat flow} $\dinexact Q_i$ is the energy transferred via microscopic degrees of freedom:
\begin{equation}
\label{eq:heat}
\dinexact Q_i := T_i\, dB_i = k_BT_i\ln 2\cdot d\log\pi_i.
\end{equation}

In our inclusive approach, there is no external driving. Instead, the \textbf{work} $\dinexact W_i$ is done by the reservoir; it is the energy transferred via changes in the macrovariables $\mathbf V_i$:
\begin{equation}
\label{eq:work}
\dinexact W_i = -\left(\frac{\partial E_i}{\partial\mathbf V_i}\right)_{B_i}\!\!\boldsymbol\cdot d\mathbf V_i.
\end{equation}

A reservoir whose only macrovariable is energy (i.e., with $d_i=0$) cannot exchange work, and is known as a \textbf{heat reservoir}. Conversely, a reservoir whose $\pi_i$ is a constant function cannot exchange heat, and is known as a \textbf{work reservoir}. Substituting \labelcref{eq:heat,eq:work} into \labelcref{eq:partials} yields the \textbf{first law of thermodynamics}
\begin{equation*}
dE_i = \dinexact Q_i - \dinexact W_i;
\end{equation*}
or after integrating over a given time interval,
\begin{equation}
\label{eq:firstlaw}
\Delta E_i = Q_i - W_i.
\end{equation}

From now on, we use the bold lowercase variable
\[\mathbf x := (x,E_1,\mathbf V_1,\ldots,E_m,\mathbf V_m)\]
to denote the joint coarse-grained state of our base and reservoir systems. It consists of the base system's mesostate $x\in\mathcal X$, which we assume to be of unit Liouville measure, along with the macrostates $(E_i,\,\mathbf V_i)$ of $m$ reservoirs. Thus, the mesostates $x$, the energies $E_i$, the macrovariables $\mathbf V_i$, and the temperatures $T_i$ are all implicitly functions of the joint state $\mathbf x$. The mixing within macrostates should be much faster than the transitions between them; this is the main physical assumption which enables us to treat $\mathbf x$ as a Markovian state.

Its joint Liouville measure
\begin{equation}
\label{eq:jointmeasuregeneral}
\pi(\mathbf x) := \prod_{i=1}^m \pi_i(E_i,\,\mathbf V_i)
\end{equation}
is stationary with respect to the dynamics $P$. Hence, the joint algorithmic entropy \labelcref{eq:entropy} expands to
\begin{equation}
\label{eq:jointentropygeneral}
S_\pi(\mathbf x\mid\widetilde P)
:= K(\mathbf x\mid\widetilde P) + \sum_{i=1}^m \log\pi_i(E_i,\,\mathbf V_i).
\end{equation}
Assuming the macrovariables are simply describable, we have
\[K(\mathbf x\mid\widetilde P)
\approx K(x\mid\widetilde P)
= S_\sharp(x\mid\widetilde P),
\qquad\log\pi_i(E_i,\,\mathbf V_i)
\approx S_{\pi_i}(E_i,\,\mathbf V_i\mid\widetilde P).\]
In other words, the joint algorithmic entropy \labelcref{eq:jointentropygeneral} is approximately the sum of the individual systems' algorithmic entropies. Motivated by these approximations, during a joint state transition $\mathbf x\rightarrow\mathbf y$, we define the base system's \textbf{entropy gain} by
\begin{equation*}
\Delta K(\mathbf x\rightarrow\mathbf y)
:= K(\mathbf y\mid\widetilde P) - K(\mathbf x\mid\widetilde P).
\end{equation*}

It can be decomposed into a sum
\begin{equation}
\label{eq:eplusi}
\Delta K(\mathbf x\rightarrow\mathbf y)
= \Delta_e K(\mathbf x\rightarrow\mathbf y)
+ \Delta_i K(\mathbf x\rightarrow\mathbf y),
\end{equation}
of the (reversible) \textbf{entropy flow} from the environment
\begin{equation}
\label{eq:entropyflow}
\Delta_e K(\mathbf x\rightarrow\mathbf y)
:= \log\pi(\mathbf x) - \log\pi(\mathbf y),
\end{equation}
and the (irreversible) \textbf{entropy production}
\begin{equation}
\label{eq:entropyproduction}
\Delta_i K(\mathbf x\rightarrow\mathbf y)
:= S_\pi(\mathbf y\mid\widetilde P)-S_\pi(\mathbf x\mid\widetilde P).
\end{equation}

This completes our core definitions. In practice, the stochastic thermodynamics literature seldom mentions the reservoir macrovariables. Instead, it labels the different means by which the base system can transition from $x$ to $y$, by \textbf{mechanisms} $i$ \citep{shiraishi2023introduction}, such that the triple $(\mathbf x,\,y,\,i)$ uniquely determines the resulting joint state $\mathbf y$, and $(x,\,y,\,i)$ determines the transition probability
\[P^{(i)}(y,\,x)
:= P(\mathbf y,\,\mathbf x).\]
Exactly one mechanism occurs per time step, and each has a different outcome, so that
\[\forall\mathbf x,\quad \sum_{i,y} P^{(i)}(y,\,x)
= \sum_{\mathbf y} P(\mathbf y,\,\mathbf x)
= 1.\]

For example, in settings where at most one reservoir changes at a time, $i$ would be the index of this reservoir. If the macrovariables of reservoir $i$ are conserved quantities, their changes are equal and opposite to those of the base system during any transition $x\rightarrow y$. Now writing
\begin{align*}
\widetilde P^{(i)}(x,\,y)
&:= \widetilde P(\mathbf x,\,\mathbf y)
= \frac{P(\mathbf y,\,\mathbf x)\pi(\mathbf x)}{\pi(\mathbf y)},
\end{align*}
the entropy flow \labelcref{eq:entropyflow} can be expressed in terms of just the base system:
\begin{equation}
\label{eq:entropyflowmechanism}
\Delta_e K(\mathbf x\rightarrow\mathbf y)
= \log\frac{\pi(\mathbf x)}{\pi(\mathbf y)}
= \log\frac{\widetilde P(\mathbf x,\,\mathbf y)}{P(\mathbf y,\,\mathbf x)}
= \log\frac{\widetilde P^{(i)}(x,\,y)}{P^{(i)}(y,\,x)}.
\end{equation}
By \labelcref{eq:eplusi,eq:entropyflowmechanism}, the entropy production is given by
\begin{align}
\label{eq:entropyproductionmechanism}
\Delta_i K(\mathbf x\rightarrow\mathbf y)
&= K(\mathbf y\mid\widetilde P) - K(\mathbf x\mid\widetilde P) + \log\frac{P^{(i)}(y,\,x)}{\widetilde P^{(i)}(x,\,y)}
\\&\approx K(y\mid\widetilde P) - K(x\mid\widetilde P) + \log\frac{P^{(i)}(y,\,x)}{\widetilde P^{(i)}(x,\,y)}\nonumber.
\end{align}

We remark that the detailed balance condition $P=\widetilde P$, which expands to \labelcref{eq:detailedbalance}, is equivalent to having $P^{(i)}=\widetilde P^{(i)}$ for all $i$, which the literature refers to as \textbf{local detailed balance} \citep{shiraishi2019fundamental,maes2021local,shiraishi2023introduction}. In this article, we do not assume (local) detailed balance. We also opt for the more explicit notation $P(\mathbf y,\,\mathbf x)$, in terms of joint states rather than mechanisms.

Having defined our key thermodynamic quantities, we now derive relationships between them, by drawing upon the mathematical results in \Cref{sec:apxrandom}. \Cref{thm:lawS} formulates the second law of thermodynamics as an \emph{integral fluctuation} inequality. It says that, regardless of the initial state distribution, the entropy production \labelcref{eq:entropyproduction} has a strong statistical tendency to be nonnegative:
\begin{equation}
\label{eq:flucintegral}
\expect{2^{-\Delta_i K}} \lmul 1.
\end{equation}

Substituting \labelcref{eq:entropyflow,eq:entropyproduction} into \labelcref{eq:indivpi,eq:indivS} yields a pair of \emph{detailed fluctuation} inequalities. They bound the entropy flow and production for every individual state transition:
\begin{align}
\Delta_e K(\mathbf x\rightarrow \mathbf y)
&\lplus \log\frac{1}{P(\mathbf y,\,\mathbf x)} - K(\mathbf x\mid \mathbf y,\widetilde P),
\label{eq:flucKe}
\\-\Delta_i K(\mathbf x\rightarrow \mathbf y)
&\lplus \log\frac{1}{P(\mathbf y,\,\mathbf x)} - K(\mathbf y\mid \mathbf x^*_{\widetilde P},\widetilde P).
\label{eq:flucKi}
\end{align}
While we obtain $\textbf x^*_{\widetilde P}:=(\textbf x,\,K(\textbf x\mid\widetilde P))$ in \Cref{sec:apxrandom} for technical reasons, in practice we can think of $\textbf x^*_{\widetilde P}$ as simply $\textbf x$, since their information content is usually about equivalent \citep[\S3.3.2]{gacs2021lecture}.

To interpret these inequalities, note that the \textbf{logical irreversibility} $K(\mathbf x\mid\mathbf y,\widetilde P)$ is the amount of information lost about a previous state $\mathbf x$, upon transitioning to $\mathbf y$. To compensate for the lost information, \labelcref{eq:flucKe} says that either $\textbf y$ must be a low-probability outcome, or else entropy must flow into the environment. Meanwhile, \labelcref{eq:flucKi} says that entropy production can only be negative for state transitions that are less likely than their ``algorithmic probability'' $2^{-K(\mathbf y\mid \mathbf x^*_{\widetilde P},\widetilde P)}$.

The inequalities \labelcref{eq:flucintegral,eq:flucKe,eq:flucKi} are information-theoretic in nature, holding for very general environments with arbitrary $\pi$. Next, we consider constant temperature environments, for which these inequalities resemble well-known thermodynamic fluctuation theorems.

\subsection{Constant temperature reservoirs}
\label{sec:fluctuation}

Suppose each reservoir's temperature $T_i$ is constant. Holding $\mathbf V_i$ fixed while integrating \labelcref{eq:partials} with respect to $B_i$ yields
\begin{equation}
\label{eq:workreservoir}
E_i = T_iB_i + E_{i,\mathrm{work}},
\end{equation}
with a ``constant of integration'' $E_{i,\mathrm{work}}(\mathbf V_i)$ that depends only on the macrovariables $\mathbf V_i$. The heat and work become exact differentials, since substituting \labelcref{eq:workreservoir} into \labelcref{eq:work} yields $\dinexact W_i = -dE_{i,\mathrm{work}}$, and then integrating yields
\begin{align}
Q_i
&= T_i\Delta B_i
= k_BT_i\ln 2\cdot \Delta\log\pi_i
= -k_BT_i\ln 2\cdot \Delta_e K,
\label{eq:heat2}
\\W_i
&= -\Delta E_{i,\mathrm{work}}.
\label{eq:work2}
\end{align}
As a result, we no longer need continuous variables to differentiate: provided that $E_i$ changes linearly with $B_i$, we can \emph{define} the heat and work using \Cref{eq:workreservoir,eq:heat2,eq:work2}.

Now, we solve \labelcref{eq:workreservoir} for
\begin{equation*}
\ln\pi_i(E_i,\,\mathbf V_i)
= \frac{B_i(E_i,\,\mathbf V_i)}{k_B}
= \frac{E_i - E_{i,\mathrm{work}}(\mathbf V_i)}{k_BT_i}.
\end{equation*}
Substituting into \labelcref{eq:jointmeasuregeneral}, the joint Liouville measure is
\begin{equation}
\label{eq:jointmeasuremany}
\pi(\mathbf x)
= \prod_{i=1}^m \pi_i(E_i,\,\mathbf V_i)
= \exp\left(\sum_{i=1}^m\frac{E_i-E_{i,\mathrm{work}}(\mathbf V_i)}{k_BT_i}\right).
\end{equation}

Up to a normalization factor, \labelcref{eq:jointmeasuremany} corresponds to a number of well-known formulas for the Gibbs measure. For example, consider the case where each reservoir $i$ has not only a constant temperature $T_i$, but also a constant pressure $p_i$ and chemical potential $\mu_i$ for a common species of particle. Then, $\mathbf V_i$ consists of the reservoir's volume $V_i$ and particle number $N_i$, and the Gibbs measure is given by \labelcref{eq:jointmeasuremany} with $E_{i,\mathrm{work}}(V_i,\,N_i):=-p_iV_i+\mu_iN_i$ \citep{guggenheim1939grand}. The mechanical term $-p_iV_i(x)$ and the chemical term $\mu_iN_i(x)$ correspond to the reservoir's ability to do each type of work.

In the case $m=1$, where there is only one reservoir, the joint state $\mathbf x$ is effectively determined as a function of only the base system state $x$. Formally, let $E_0(x)$, $V_0(x)$, and $N_0(x)$ respectively denote the energy, volume, and particle number of the base system at state $x$, and suppose that the totals  $E_0(x)+E_1$, $V_0(x)+V_1$, and $N_0(x)+N_1$ are conserved. After normalizing the reservoir macrovariables $E_1,V_1,N_1$ so that each of the totals is zero,
\begin{equation*}
\mathbf x = (x,\,E_1,\,V_1,\,N_1) = (x,\,-E_0(x),\,-V_0(x),\,-N_0(x)),
\end{equation*}

Consequently, \labelcref{eq:jointmeasuremany} simplifies to
\begin{equation}
\label{eq:jointmeasureone}
\pi(\mathbf x)
= \pi(x)
= \exp\left(\frac{-E_0(x)-E_\mathrm{work}(x)}{k_BT}\right),
\end{equation}
with $E_\mathrm{work}(x) := pV_0(x)-\mu N_0(x)$ in the case of constant pressure $p$ and chemical potential $\mu$.

More generally, we consider any kind of single-reservoir environment at constant temperature $T$, whose macrovariables are determined as functions of $x$. \labelcref{eq:jointmeasureone} still holds, with a possibly different potential function $E_\mathrm{work}(x)$. As a result, the joint algorithmic entropy is given by
\begin{equation}
\label{eq:jointentropyone}
S_\pi(\mathbf x\mid\widetilde P)
\eqplus S_\pi(x\mid\widetilde P)
= K(x\mid\widetilde P) - \frac{E_0(x)+E_\mathrm{work}(x)}{k_BT\ln 2}.
\end{equation}
It is customary to multiply aggregate entropies by $-k_BT\ln 2$, in order to express them in units of energy. The result is the \textbf{total algorithmic free energy}
\begin{equation}
\label{eq:gibbs}
G(x) := E(x) + E_\mathrm{work}(x) - K(x\mid\widetilde P)\cdot k_BT\ln 2.
\end{equation}

$G$ serves as a convenient accounting mechanism. Despite being a function of only the base system's state $x$, it tracks the total entropy production:
\begin{equation}
\label{eq:deltaG}
\Delta G \eqplus -k_BT\ln2\cdot\Delta_i K.
\end{equation}
As such, $G$ is non-increasing up to fluctuations. To be precise, \labelcref{eq:flucintegral,eq:deltaG} imply
\begin{equation}
\label{eq:flucG}
\expect{\exp\left(\frac{\Delta G}{k_BT}\right)}
= \expect{2^{\Delta G/k_BT\ln 2}}
\lmul 1.
\end{equation}

It is also useful to consider thermodynamic potentials which track only \emph{some} changes in entropy; we interpret them as \emph{resources} that convert to and from the excluded form(s) of entropy. For example, define the \textbf{Helmholtz algorithmic free energy} by
\begin{equation}
\label{eq:helmholtz}
F(x) := E(x) - K(x\mid\widetilde P)\cdot k_BT\ln 2.
\end{equation}

During any state transition, \labelcref{eq:work2,eq:gibbs,eq:helmholtz} imply
\begin{equation}
\label{eq:FplusW}
\Delta G
= \Delta F + \Delta E_\mathrm{work}
= \Delta F - W.
\end{equation}

Substituting \labelcref{eq:FplusW} into the integral fluctuation inequality \labelcref{eq:flucG} yields
\begin{equation}
\label{eq:jarzynski}
\expect{\exp\left(\frac{\Delta F- W}{k_BT}\right)}
\lmul 1.
\end{equation}

Markov's inequality (or \labelcref{eq:lawS}) then implies,  with probability greater than $1-\delta$,
\[\Delta F - W \lplus k_BT\ln\frac 1\delta.\]
Since the right-hand side is negligible at macroscopic scales, we see that $F$ measures the base system's capacity for work: free energy must be spent in order for the base system to do work; and conversely, work must be done on the base system in order to replenish free energy.

To get a closer view of the fluctuations in \labelcref{eq:jarzynski}, substitute \labelcref{eq:deltaG,eq:FplusW} into the detailed fluctuation inequality \labelcref{eq:flucKi}:
\[\frac{\Delta F - W}{k_BT\ln 2}
\lplus \log\frac{1}{P(y,\,x)} - K(y\mid x^*_{\widetilde P},\widetilde P).\]

Aside from the $\lmul$ sign, \labelcref{eq:jarzynski} is symbolically identical to \citet{jarzynski2011equalities}'s equality. However, there are important differences between the two results: our algorithmic free energy $F$ is a trajectory-level quantity, defined as a function of individual states rather than ensembles. Moreover, we take an inclusive view of the work $W$, as an interaction with the reservoir rather than as external driving.

Finally, to find an explicit connection between information and heat transfer, we change our choice of thermodynamic potential, from $F$, to the internal entropy $K$. Using \labelcref{eq:eplusi,eq:heat2},
\[\Delta_i K
= \Delta K - \Delta_e K
= \Delta K + \frac{Q}{k_BT\ln 2}.\]
Substituting into the integral fluctuation inequality \labelcref{eq:flucintegral} yields
\begin{equation}
\label{eq:landauer}
\expect{2^{-\left(\Delta K + \frac{Q}{k_BT\ln 2}\right)}}
\lmul 1.
\end{equation}
By Markov's inequality (or \labelcref{eq:lawS}) again, with probability greater than $1-\delta$,
\begin{equation}
\label{eq:landauer2}
\Delta K + \frac{Q}{k_BT\ln 2} \gplus -\log\frac 1\delta.
\end{equation}

In order for a digital memory to clear its data, \citet{landauer1961irreversibility} argued that it must emit at least $k_BT\ln 2$ of heat per erased bit. Since clearing data reduces its description complexity $K$, we can think of \labelcref{eq:landauer,eq:landauer2} as mathematically rigorous formulations of Landauer's principle. The impact of heat flow on the energy efficiency of computer hardware depends on the extent to which the flow is reversible; we examine this in \Cref{sec:landauer}.

The fluctuations in \labelcref{eq:landauer} are described by the inequality \labelcref{eq:flucKe}, which generalizes the earlier bounds of \citet{zurek1989thermodynamic} and \citet{kolchinsky2023}. In the case of near-deterministic transitions with negligible complexity (i.e., $\log P(y,\,x)\eqplus K(\widetilde P)\eqplus 0$), \labelcref{eq:flucKe} reduces to Zurek's inequality
\begin{equation}
\label{eq:zurek}
-\Delta_e K(x\rightarrow y) \gplus K(x\mid y).
\end{equation}

On the other hand, substituting \labelcref{eq:heat2} into \labelcref{eq:flucKe} yields Kolchinsky's inequality
\begin{equation}
\label{eq:kolchinsky}
\frac{Q(x\rightarrow y)}{k_BT\ln 2}
\gplus K(x\mid y,\widetilde P) - \log\frac{1}{P(y,\,x)}.
\end{equation}
It can be seen as a detailed version of Landauer's principle, giving the minimum heat transfer that accompanies a state transition $x\rightarrow y$, in terms of its probability $P(y,\,x)$ and logical irreversibility $K(x\mid y,\widetilde P)$.

\subsection{Refining the second law}
\label{sec:refine}

Now, we return to the fully general setting from \Cref{sec:algoentropy}, to deal with the algorithmic entropy's dependence on the conditional parameters $\widetilde P$. In most applications, the information content of $\widetilde P$ consists of a constant part and a variable part. For example, suppose we want to study a particular time-homogeneous $\pi$-stochastic Markov chain, at many different times. Then, \labelcref{eq:iterate} determines the transition matrix $P_{\Delta t}$ in terms of a constant part $P_1$ and a variable part $\Delta t$; to compute its dual $\widetilde P_{\Delta t}$, we add $\pi$ to the constant part.

We regard $(P_1,\,\pi)$ as built into the fundamental laws of physics, and assume that they have short encodings on the ``natural computers'' of the Universe. This assumption can be viewed as a complexity-theoretic version of the physical Church-Turing thesis \citep{kolchinsky2020thermodynamic,wolpert2024implications}, essentially saying that the Universe can implement its own laws on a small computer. Formally, we model this by choosing our reference universal computer in such a way that $K(P_1,\,\pi)\approx 0$. Only the variable part (in this case, $\Delta t$) requires an explicit correction.

First, we prove the correction for general $\widetilde P$; later, we consider the case where only $\Delta t$ is variable. To keep this subsection brief, we apply it only to the tail bound \labelcref{eq:lawS}, though the same correction can also be applied to the integral bound \labelcref{eq:flucS} and the detailed bound \labelcref{eq:flucKi}. A reader who is less interested in mathematical details may skip to the main result, \Cref{cor:secondlaw}.

\begin{theorem}
\label{thm:lawunconditional}
Let $\pi:\mathcal X\rightarrow\reals^+\setminus\{0\}$ be a measure and $X,Y$ be $\mathcal X$-valued random variables, such that the matrix $P(y,\,x) := \Pr(Y=y \mid X=x)$ is $\pi$-stochastic with a computable dual $\widetilde P$. Let $\delta > 0$. Then, with probability greater than $1-\delta$,
\begin{align}
S_\pi(X) - S_\pi(Y)
&\lplus I(X\mi\widetilde P) - I(Y\mi\widetilde P) + \log\frac 1\delta \nonumber
\\&\lplus K(\widetilde P) + \log\frac 1\delta. \label{eq:lawunconditional}
\end{align}
\end{theorem}

\begin{proof}
For all $x\in\mathcal X$, the defining equations \labelcref{eq:entropy,eq:mutinf} imply

\begin{align}
S_\pi(x\mid\widetilde P,K(\widetilde P)) \nonumber
&= \log\pi(x) + K(x\mid\widetilde P,K(\widetilde P)) \nonumber
\\&\eqplus \log\pi(x) + K(x) - I(x\mi\widetilde P) \nonumber
\\&= S_\pi(x) - I(x\mi\widetilde P). \label{eq:correction}
\end{align}

Now, fix $\delta>0$. We condition \Cref{thm:lawS} on the additional data $K(\widetilde P)$ to find that, with probability greater than $1-\delta$,
\[S_\pi(X\mid\widetilde P,K(\widetilde P))-S_\pi(Y\mid\widetilde P,K(\widetilde P)) \lplus \log\frac 1\delta.\]

In this event, \labelcref{eq:correction} yields
\begin{align*}
S_\pi(X) - S_\pi(Y)
&\eqplus S_\pi(X_s\mid\widetilde P,K(\widetilde P)) + I(X\mi\widetilde P) - S_\pi(Y\mid\widetilde P,K(\widetilde P)) - I(Y\mi\widetilde P)
\\&\lplus I(X\mi\widetilde P) - I(Y\mi\widetilde P) + \log\frac 1\delta.
\end{align*}

\labelcref{eq:lawunconditional} now follows from the general bounds $0 \lplus I(x\mi z) \lplus K(z)$.
\end{proof}

Next, we show that \Cref{thm:lawunconditional} is tight: in the deterministic and reversible case, its conclusion holds with equality. To be precise, consider the case where $P$ is a permutation matrix; equivalently, the dynamics are described by a bijective transformation $f:\mathcal X\rightarrow\mathcal X$.

\begin{corollary}
\label{cor:janzing}
For all computable bijections $f:\mathcal X\rightarrow\mathcal X$, and $x\in\mathcal X$,
\begin{align*}
K(f(x)) - K(x)  &\eqplus I(f(x)\mi f) - I(x\mi f).
\end{align*}

In particular, if $I(x\mi f) \eqplus 0$, then
\begin{equation}
\label{eq:janzing}
0 \lplus K(f(x)) - K(x) \lplus K(f).
\end{equation}
\end{corollary}

\begin{proof}
Define a permutation matrix $P$ as follows: for $x,y\in\mathcal X$, let $P(y,\,x):=1$ if $y=f(x)$, and $P(y,\,x):=0$ otherwise. Since $P$ is doubly stochastic, $\widetilde P$ is its transpose, computable by a constant-length program together with $f$. In order to apply \Cref{thm:lawunconditional}, let $\pi:=\sharp$ so that $S_\pi=K$, let the ``random'' variable $X$ be a constant $x\in\mathcal X$ with probability one, and let $Y := f(x)$. Clearly, if the probability of a non-random event is positive, then it occurs with certainty. Therefore, by setting $\delta:=1/2$ in \Cref{thm:lawunconditional},
\[K(x) - K(f(x)) \lplus I(x\mi f) - I(f(x)\mi f) \lplus K(f).\]

Repeating the same argument for the inverse function yields
\[K(f(x)) - K(x) \lplus I(f(x)\mi f) - I(x\mi f) \lplus K(f).\]

Combining these inequalities yields the desired conclusions.
\end{proof}

The leftmost inequality of \labelcref{eq:janzing} first appeared in \citet{janzing2016algorithmic}. There, the implication
\[I(x\mi f)\eqplus 0 \implies K(x)\lplus K(f(x))\]
was interpreted as saying that the second law of thermodynamics (increase in $K$) is due to algorithmic independence of the initial condition $x$ from the dynamical law $f$. The problem with their interpretation lies in the rightmost inequality of \labelcref{eq:janzing}: $f$ would have to be extraordinarily complex to allow entropy production at a physically meaningful scale. For any deterministic dynamical law that we can feasibly write down, \Cref{cor:janzing} really says that the entropy cannot change by a physically meaningful amount (recall the unit conversions \labelcref{eq:conversion}). Thus, randomness (which in classical physics comes from coarse-graining) is necessary for substantial entropy production to occur.

That being said, we can offer a useful interpretation in line with that of \citet{janzing2016algorithmic}. A doubly stochastic law $P$ can be implemented by choosing a bijection $f$ at random \citep{revesz1962probabilistic}, e.g., by repeated tosses of a fair coin. The combination of a low-complexity $P$, along with the results of sufficiently many coin tosses, then serves as a deterministic law $f$ of very high complexity. Since the coins contribute the bulk of $f$'s information content, algorithmic independence from $f$ really means independence from the coin tosses. If the coins are indeed independent of $x$, \Cref{cor:janzing} implies that the entropy cannot decrease, and may increase by up to as many bits as there are coins.

In the setting of a time-homogeneous process, we want to compare the entropies $S_\pi(X_s)$ and $S_\pi(X_t)$, at times $s<t$. The transition matrix $P_{t-s}$ governing the state transition is generated by either $P_1$ (if time is discrete), or a rate matrix (if time is continuous). A programmatic description of our process's physics should compute the generator; and from it, the transition matrix $P_{t-s}$ and its dual $\widetilde P_{t-s}$, over any desired time interval $[s,\,t]$.

Formally, we postulate the existence of a short computer program $p$, such that for all $\Delta t$ and $x,y\in\mathcal X$, our reference computer $U$ outputs $U(p,\,\Delta t,\,y,\,x) = \widetilde P_{\Delta t}(y,\,x)$. In other words, the pair $(p,\,\Delta t)$ is a program for $\widetilde P_{\Delta t}$. In the continuous-time case, since the uncountable set $\reals^+$ has no encoding, we only consider rational durations $\Delta t\in\mathbb Q^+$. By correcting for the description complexity of $\Delta t$, we arrive at our most comprehensive, duration-dependent second law of thermodynamics.

\begin{corollary}[Algorithmic second law of thermodynamics]
\label{cor:secondlaw}
Let $\pi:\mathcal X\rightarrow\reals^+\setminus\{0\}$ be a measure, $(X_t)_{t\in\mathcal T}$ be a stochastic process in either continuous ($\mathcal T=\reals^+$) or discrete ($\mathcal T=\ints^+$) time, and fix $p\in\bits^*$ so that $|p|\eqplus 0$. Consider a pair $s,t\in\mathcal T$ with $t-s\in\mathbb Q^+$, such that the matrix $P(y,\,x):=\Pr(X_t = y \mid X_s = x)$ is $\pi$-stochastic, and its dual satisfies $\widetilde P(y,\,x)=U(p,\,t-s,\,y,\,x)$.
Then, for $\delta>0$, with probability greater than $1-\delta$,
\begin{align}
S_\pi(X_s) - S_\pi(X_t)
&\lplus I(X_s\mi t-s) - I(X_t\mi t-s) + \log\frac 1\delta \nonumber
\\&\lplus K(t-s) + \log\frac 1\delta.\label{eq:secondlaw}
\end{align}
\end{corollary}

\begin{proof}
Using the pair $(p,\,t-s)$ to encode $\widetilde P$, the conclusion of \Cref{thm:lawunconditional} becomes
\begin{align*}
S_\pi(X_s) - S_\pi(X_t)
&\lplus I(X_s\mi(p,\,t-s)) - I(X_t\mi(p,\,t-s)) + \log\frac 1\delta
\\&\lplus K(p,\,t-s) + \log\frac 1\delta.
\end{align*}
Since $|p|\eqplus 0$, \labelcref{eq:secondlaw} follows.
\end{proof}

We remark that both of the fluctuation terms, $K(t-s)$ and $\log(1/\delta)$, are necessary. The former allows periodic visits to low-entropy states, such as in a deterministic process with a very long cycle; whereas the latter allows chance encounters with low-entropy states, such as in a random mixing process. The Poincar\'e recurrence theorem famously predicts that Hamiltonian systems eventually return to states of low entropy \citep{saussol2009introduction}; \Cref{cor:secondlaw} is consistent with this finding \footnote{That said, entropy fluctuations at the scale of Poincar\'e recurrence are incredibly rare. To illustrate, consider $10^{10^{20}}$ consecutive time steps of any simply describable duration, be they seconds or years. If we set $\delta := \left(10^{10^{20}}\right)^{-3}$ in \Cref{cor:secondlaw}, then by a union bound over all pairs $(s,t)$ of these times, the probability of \labelcref{eq:secondlaw} failing for even \emph{one} pair is less than one in $10^{10^{20}}$. Since $10^{10^{20}} < 2^{3.33\times 10^{20}}$, \labelcref{eq:secondlaw} bounds the largest entropy decrease by
\begin{align*}
K(t-s) + \log\frac{1}{\delta}
&< 3.33\times 10^{20} + 2\log(3.33\times 10^{20}) + 9.99\times 10^{20}
< 1.333\times 10^{21} \text{ bits}
< \SI{0.013}{\joule\per\kelvin}.
\end{align*}
Each increment of the topmost exponent in $10^{10^{20}}$ would only multiply this bound tenfold.
}.

For realistic systems that do not change too quickly, the condition $t-s\in\mathbb Q^+$ is not a serious limitation: a small change in $t$ suffices not only to make $t-s$ rational, but also to make its complexity $K(t-s)$ negligible. For concreteness, suppose we restrict our comparisons of entropy to durations that are multiples of Planck's time. In Planck units, since the age of the Universe is less than $2^{203}$, all durations up to the present can be represented by integers in the range $1 < \Delta t < 2^{203}$, for which
\[K(\Delta t)
\lplus \log\Delta t + 2\log\log\Delta t
< 203 + 2\cdot 8 \text{ bits}
= 219 \text{ bits}
< \SI{2.1e-21}{\joule\per\kelvin}.\]
Therefore, up to negligible fudge terms and exceedingly rare fluctuations, \Cref{cor:secondlaw} says that the unconditional entropy $S_\pi$ is non-decreasing over time.

\section{Applications}
\label{sec:discussion}

\subsection{Inclusive dynamics and open systems}
\label{sec:discussionintro}

Our setup thus far is quite general, but cumbersome for the purpose of constructing examples. We focused on inclusive closed system models, in which all influences are explicitly accounted for. As a result, the dynamics are fully determined by the laws of physics, which we assume to be measure-preserving (hence, $\pi$-stochastic on a Markovian coarse-graining), simply describable, and time-homogeneous. \Cref{cor:secondlaw} applies directly to such settings.

On the other hand, sometimes we want to omit the details of some influences, leaving an \textbf{open system} model. For example, in the setting of \Cref{sec:environment,sec:fluctuation}, if we model only the base system without the environment's reservoirs, this is an open system. Its mesostates have constant measure, and yet the dynamics have a non-constant stationary measure $\pi$. Since open systems can trade with their environment, measure-preservation and conservation laws do not directly apply.

Without modeling the environment explicitly, can we correct for it to say anything useful? To derive formulas for the entropy flow and production, we must start from the corresponding inclusive model and then eliminate the environment measure $\pi$. We did this in \labelcref{eq:entropyflowmechanism,eq:entropyproductionmechanism}; these formulas still depend on $\pi$ through the dual matrix, but that too is eliminated in the case of (local) detailed balance. We can then apply results such as \Cref{thm:lawunconditional} and \Cref{cor:secondlaw}, by substituting the correct formula \labelcref{eq:entropyproductionmechanism} for entropy production.

Other open systems may experience algorithmically complex or time-dependent dynamics. Again, we begin with the inclusive model, and then try to eliminate the influence. Consider a \emph{control system} with state space $\mathcal C$, that influences a \emph{base system} with state space $\mathcal X$. Suppose the control state stays still, i.e., the joint transition matrix $P:(\mathcal C\times\mathcal X)\times(\mathcal C\times\mathcal X)\rightarrow\reals^+$ satisfies
\[c\ne d \implies P((d,y),\,(c,x))=0.\]
For any coarse-grained Liouville measure
\[\pi(c,x):=\pi_{\mathcal C}(c)\pi_{\mathcal X}(x)\quad\text{for }c\in\mathcal C,\;x\in\mathcal X,\]
it is easy to verify that $P$ is $\pi$-stochastic iff each of the submatrices $P((c,\cdot),\,(c,\cdot))$ are $\pi_{\mathcal X}$-stochastic.

In this manner, a single \emph{global dynamics} $P$ can emulate a wide variety of \emph{local dynamics} $P((c,\cdot),\,(c,\cdot))$ on our base system. Since the submatrices are $\pi_{\mathcal X}$-stochastic, results such as \Cref{thm:lawunconditional} and \Cref{cor:secondlaw} apply to our base system. However, the submatrix is not fully determined by the laws of physics $P$; it also depends on the control input $c$. If $K(c)$ is sufficiently small, it is safe to ignore this dependence.

However, if $K(c)$ is large, the control might represent a complex piece of information, such as the mind of a Maxwell's demon. In the next subsection, we revisit the famous thought experiment using our algorithmic inequalities. If we insist on modeling $x$ as an open system, then the algorithmic entropy must be conditioned on $c$.

A useful variant of a control is a reversible program counter or \emph{clock} $\mathcal C := \ints_m$, with $\pi_{\mathcal C} := \sharp$. Instead of staying still, it ticks ahead in a predetermined manner:
\[c+1\not\equiv d \pmod{m} \implies P((d,y),\,(c,x))=0.\]
Although the global dynamics $P$ is time-homogeneous, a clock enables the base system to undergo a sequence of different $\pi_{\mathcal X}$-stochastic transitions, given by the submatrices $P((c+1,\cdot),\,(c,\cdot))$. If $m$ is not too large, then $K(c)\lplus 2\log m$ is quite small.

We can apply \Cref{thm:lawunconditional} to systems with time-dependent dynamics. Consider a time interval $[s,\,t]$, during which the evolution is given by a short sequence of simply describable time-dependent transition matrices. Then, setting $(X,Y):=(X_s,X_t)$ in \Cref{thm:lawunconditional}, we have $K(\widetilde P)\eqplus 0$. Allowing for a small but constant probability $\delta$ of failure, its conclusion simplifies to $S_{\pi_{\mathcal X}}(X_s) \lplus S_{\pi_{\mathcal X}}(X_t)$. Recall from \Cref{sec:algoentropy} that, for non-reservoir systems, we coarse-grain into equal-sized mesostates. Therefore, in the absence of heat transfer, the dynamics are doubly stochastic, the algorithmic entropy is $S_\sharp=K$, and the theorem's conclusion becomes
\begin{equation}
\label{eq:secondlawsimplified}
K(X_s)\lplus K(X_t).
\end{equation}

It may also seem cumbersome to construct interesting examples of transition matrices. Fortunately, by R\'ev\'esz's generalization of the Birkhoff-von Neumann theorem, transitioning by a doubly stochastic matrix is equivalent to transitioning by a probabilistic mixture of deterministic bijections \citep{revesz1962probabilistic}. This means, instead of writing an explicit transition matrix $P$, we can describe the dynamics as a \emph{random bijection} $x\mapsto F(x)$, where the bijection $F:\mathcal X\rightarrow\mathcal X$ is sampled independently at each time step, from some distribution with low description complexity \footnote{Strictly speaking, we should write $F:\Omega\times\mathcal X\rightarrow\mathcal X$ to express dependence on an ambient probability space $\Omega$. If $\mathcal X$ is countably infinite, then there are uncountably many bijections on it, so the probability measure need not be discrete. Instead, it can be represented by a uniformly computable sequence of functions $\Gamma_n:\mathcal X^{\mathcal X_n}\rightarrow\reals$, where $\mathcal X_n\subset\mathcal X$ consists of the states whose encodings have length at most $n$, and $\mathcal X^{\mathcal X_n}$ is the countable set of functions $\tilde f:\mathcal X_n\rightarrow\mathcal X$. We take $\Gamma_n(\tilde f)$ to be the probability of sampling a function $f:\mathcal X\rightarrow\mathcal X$ whose restriction to $\mathcal X_n$ is $\tilde f$. In order to compute the transition matrix entry $P(y,\,x)$ from $\Gamma$, let $n=|x|$, and enumerate functions $\tilde f:\mathcal X_n\rightarrow\mathcal X$ until their total probability is as close to $1$ as desired. Then, $P(y,\,x)$ is approximately the probability assigned to functions that satisfy $\tilde f(x) = y$. Since $K(P)\lplus K(\Gamma)$, if $\Gamma$ is simply describable, then so is $P$.}. If we only care to specify $F$ on a proper subset of $\mathcal X$, then it can be a random \emph{injection} whose domain and range have complements of equal cardinality, since these can be extended to bijections on all of $\mathcal X$.

For example, on the two-element state space $\mathcal X:= \bits = \{\mathtt 0,\mathtt 1\}$, there are exactly two bijections: identity and negation. Therefore, these are the only deterministic dynamics permitted on $\bits$. The full set of permitted dynamics (in the absence of heat transfer) are the mixtures of identity and negation, parametrized by a probability of negation $\alpha\in [0,1]$. Injections on proper subsets of $\mathcal X$ are also allowed: for example, we can specify that $\mathtt 0$ maps to $\mathtt 1$, without caring what $\mathtt 1$ maps to (though in this case, negation is the only bijective extension).

In summary, the behavior of open systems can be influenced in a variety of ways. Complex controls require careful accounting, but simple controls and clocks just extend our modeling capabilities: they allow us to consider systems that evolve by sequences of deterministic or random bijections on $\mathcal X$, or injections on subsets of $\mathcal X$. The transition function at each time step is sampled from a simply describable distribution. If not too many time steps are taken, then \labelcref{eq:secondlawsimplified} holds with a high probability. This time-dependent setting is flexible enough to illustrate a number of phenomena regarding the thermodynamics of information.

\subsection{Maxwell's demon}
\label{sec:maxwell}

As a warmup to more difficult examples, we now review the original challenge to the second law \citep{maruyama2009colloquium}. The core ideas here are not new, but we hope that our simple abstract presentation lends some pedagogical clarity.

In the famous thought experiment, \textbf{Maxwell's demon} has a memory that starts in a low-entropy ``clear'' state $0\in\mathcal C$. It interacts with a base system that starts in some high-entropy mesostate $x\in\mathcal X$. It is helpful to begin with a stylized special case, in which $\mathcal C=\mathcal X$ and the demon is able to reversibly perform a complete measurement, copying the system's state into memory:
\[(0,\, x) \mapsto (x,\, x).\]

Using its measurement as a control, the demon proceeds to reversibly erase the base system's entropy:
\[(x,\, x) \mapsto (x,\, 0).\]

Both of these mappings are injective on their respective domains, $\{(0,\,x):\,x\in\mathcal X\}$ and $\{(x,\,x):\,x\in\mathcal X\}$. Therefore, they can be extended to bijections on $\mathcal X\times\mathcal X$. Since they are deterministic, \Cref{cor:janzing} implies the total entropy cannot change substantially. Indeed,
\[K(0,\,x)\eqplus K(x,\,x)\eqplus K(x,\,0)\eqplus K(x).\]

The net effect is that the base system's information content is moved into the demon's memory. The erasure step uses a high-entropy control $x$; as discussed in \Cref{sec:discussionintro}, its entropy must be included in the total, for otherwise the base system's transition $x \mapsto 0$ would appear to violate the second law. The erasure is permitted precisely because it occurs in the presence of a copy. The second law forbids the demon from clearing the last copy remaining in its memory:
\[(x,\, 0) \mapsto (0,\, 0),\]
as can also be seen by noting that this map is either not injective (if defined to work for all $x$), or not simply describable (if tailored for a specific $x$).

Now we present the general case, which includes nearly all physical models of Maxwell's demon from the literature. Suppose the demon performs a partial measurement $m(x)$, where $m$ is a (possibly random, not necessarily bijective) function, whose distribution has low description complexity. After obtaining the measurement, the demon uses it as a control to transition the system from $x$ to some new (possibly random) state $y$:
\[(0,\, x) \mapsto (m(x),\, x) \mapsto (m(x),\, y).\]

The specifics of $y$'s computation are not important: as long as it amounts to a simply describable mixture of bijections, the second law expressed by \labelcref{eq:secondlawsimplified} holds with high probability. It expands to
\[K(m(x),\,x) \lplus K(m(x),\,y).\]
Subtracting $K(m(x))$ from both sides yields
\begin{equation}
\label{eq:demon}
K(x) - I(m(x)\mi x)
\eqplus K(x\mid m(x)^*) \lplus K(y\mid m(x)^*)
\eqplus K(y) - I(m(x)\mi  y).    
\end{equation}

We can interpret the conditional complexities as subjective entropies, from the point of view of a demon that knows the measurement $m(x)$. The measurement makes the base system's subjective entropy less than its objective (i.e., unconditional) entropy.

Now, \labelcref{eq:demon} implies $K(x) - K(y) \lplus I(m(x)\mi x)$. Consider the case where the mutual information of measurement is reversibly erased: ``reversibly'' meaning $K(m(x),\,x)\eqplus K(m(x),\,y)$, and ``erased'' meaning $I(m(x)\mi  y)\eqplus 0$. Then, in fact,
\[K(x) - K(y) \eqplus I(m(x)\mi  x).\]

Thus, although the second law forbids a decrease in the \emph{total} entropy, it permits the measured system to lose as much entropy as was measured from it! There is no contradiction here: since the algorithmic entropy is subadditive, it is possible to have simultaneously $K(x) \gg K(y)$ and $K(m(x),\,x)\lplus K(m(x),\,y)$.

Note that our analysis does not require the completion of a cycle, nor any ad hoc extension of the definition of entropy. The algorithmic entropy naturally accounts for both the demon's memory and the base system, satisfying the second law  of thermodynamics at every step of their evolution.

\subsection{Landauer's principle}
\label{sec:landauer}

In \Cref{sec:fluctuation}, we saw that a system can dump its unwanted entropy into a reservoir as heat. \citet{landauer1961irreversibility} first discovered this in the context of computer circuitry, arguing that logically irreversible computations necessarily convert some energy into waste heat. \citet{neyman1966negentropy} followed up with a larger bound in settings involving irreversible thermal equilibration. There are now many modern references treating Landauer's principle \citep{sagawa2014thermodynamic,frank2018physical,wolpert2019stochastic}.

We can develop similar ideas in terms of the algorithmic entropy. Just as \labelcref{eq:eplusi} decomposes a base system's entropy gain into reversible and irreversible parts, we can rearrange \labelcref{eq:eplusi} to decompose the \emph{environment}'s entropy gain into reversible and irreversible parts:
\begin{equation}
\label{eq:eplusireversed}
-\Delta_e K(\mathbf x\rightarrow\mathbf y)
= \Delta_i K(\mathbf x\rightarrow\mathbf y)
- \Delta K(\mathbf x\rightarrow\mathbf y).
\end{equation}

By \labelcref{eq:heat2}, the left-hand side is directly proportional to the heat flow. Its irreversible part is the algorithmic entropy production, or \textbf{EP cost} $\Delta_i K$. The reversible part is the drop in base system entropy $-\Delta K$; by analogy to its ensemble-based analogue in the stochastic thermodynamics literature \citep{wolpert2019stochastic,wolpert2024stochastic}, we call it the algorithmic \textbf{Landauer cost}. Thus, heat flow is directly proportional to the sum of EP and Landauer costs.

\Cref{eq:eplusireversed} helps to clarify some common misconceptions regarding Landauer's principle \citep{wolpert2019stochastic}. For example, logical irreversibility need not result in a Landauer cost, nor in heat flow (except in the deterministic case, where \labelcref{eq:zurek} holds). \citet{sagawa2012thermodynamics} demonstrates this with a physical example, though a simpler example suffices: consider a finite-state Markov chain with the transition matrix $P(y,\,x):=1/|\mathcal X|$. Each iteration is logically irreversible, overwriting the previous value with an independent uniformly distributed value. Nonetheless, the algorithmic entropy usually stays near $\log|\mathcal X|$, so the Landauer cost is zero. Since $P$ is doubly stochastic, it can be implemented without contacting a reservoir, so the heat flow is also zero.

Another common mistake is to identify the Landauer cost with a system's long-term energy consumption. In reality, Landauer costs are reversible: for a computer memory whose entropy is bounded from both above and below, positive and negative Landauer costs must approximately balance each other in the long run. That is, we have the long-run homeostasis condition $\Delta K\approx 0$, which by \labelcref{eq:eplusireversed} implies $-\Delta_e K \approx \Delta_i K$. By \labelcref{eq:heat2}, the long-term energy consumption is therefore proportional to the irreversible EP cost:
\begin{equation}
\label{eq:landauersimple}
Q \approx k_BT\ln 2\cdot \Delta_i K.
\end{equation}

This suggests an interesting accounting trick. The obvious way to compute the net heat flow $Q$ over a long series of events, is to sum the heat emission or absorption from each individual event; such a sum may include redundant Landauer costs that cancel due to opposite signs. Alternatively, \labelcref{eq:landauersimple} says that we can sum the EP costs of the individual events, and divide by $k_BT\ln 2$. The latter methodology not only avoids redundant cancellations, but also expresses the energy cost directly in terms of the algorithmic information-theoretic quantity $\Delta_i K$.

To gain some further intuition, we now examine three common types of information process: randomization, computation, and measurement. In each case, we ask whether there is an EP cost, and how that translates to net heat emission.

First, consider the generation of random data, perhaps to serve as a seed for a randomized computation. $K$ increases, corresponding to a negative Landauer cost. In principle, \citet{bennett1982thermodynamics} shows that entropy can be reversibly extracted from a heat reservoir, cooling it while randomizing a piece of digital memory. Later, the memory's entropy can be reversibly returned to the reservoir for a positive Landauer cost, warming it while clearing the memory. This cycle has zero net heat flow. Notice that the cooling step is essential: if the memory collects entropy in an uncontrolled manner, without cooling the reservoir, then \labelcref{eq:eplusireversed} requires the negative Landauer cost to be offset by a positive EP cost. By the time the memory is cleared, there will be some net heat emission, as predicted by \labelcref{eq:landauersimple}.

In the second case, consider a long string $x$ resulting from a deterministic computation. Although $x$ may appear complex, in reality its entropy $K(x)$ is about as small as the program that computed it. Only when we ignore the origins of $x$ and toss it into a stochastic reservoir, is entropy produced. Since the string was not truly random, we have $\Delta K\approx 0$. Thus, the warming of the reservoir cannot be attributed to Landauer cost, and is in fact an EP cost. In principle, a \textbf{reversible computer} can avoid EP and heat emission, clearing $x$ by running its computation in reverse \citep{Vitnyi2005TimeSA,demaine2016energy,morita2017theory,baumeler2019free}.

Finally, consider sensing or measurement. In \Cref{sec:maxwell}, we saw that a memory can reversibly take a measurement of another object. By interacting again with the same object, the measurement can be undone at zero cost. On the other hand, if either copy of the data is lost \emph{without} them interacting, either because the memory is overwritten or the source object changes state, then there is an EP cost.

To see this, consider two arbitrary systems (e.g., a memory and some object), in the respective states $x$ and $y$. Using \labelcref{eq:mutinf}, their total entropy decomposes as
\begin{equation*}
K(x,\,y) = K(x) + K(y) - I(x\mi y).
\end{equation*}
As a result, the entropy production is a sum of three terms:
\begin{equation}
\label{eq:jointproduction}
\Delta_i K := \Delta K(x,\,y)
= \Delta K(x) + \Delta K(y) + \Delta (-I(x\mi y)).
\end{equation}

Substituting into \labelcref{eq:flucintegral} (which again follows from \Cref{thm:lawS}), we obtain a formulation of the second law that explicitly accounts for the change in mutual information between the systems. It is an algorithmic analogue of the result by \citet{sagawa2012fluctuation}.

Now, suppose the two systems do not interact, making each of them an isolated system. Then, \Cref{thm:lawS} applies to each system individually: up to minor fluctuations, it says that their respective entropies $K(x)$ and $K(y)$ are non-decreasing. \Cref{thm:lawI} also applies, saying that $I(x\mi y)$ is non-increasing. Since all three terms in the decomposition \labelcref{eq:jointproduction}  are nonnegative, we can view each of them as separate EP costs. In particular, we conclude that for non-interacting systems, discarded mutual information is a form of EP.

Whether our aim is to randomize, to compute, or to measure, the absence of entropy is a resource to be carefully managed. In the first case, it is exchanged with a heat reservoir; in the second, it is encrypted by a computation; and in the third, it is stored in the mutual information between two systems. In principle, all these manipulations can be done reversibly, at zero net cost. However, when there is a mismatch between our technological mechanism and the information that it processes, then we pay an EP cost, which ultimately turns to heat according to \labelcref{eq:landauersimple}.

\subsection{An information engine}
\label{sec:engine}

In the physical world, energy is conserved. When a system ``consumes'' energy, the total energy does not decrease; instead, it transforms into waste heat. \labelcref{eq:landauersimple} equates the heat $Q$ with the entropy production $\Delta_i K$; thus, the ``resource'' that is consumed is in fact the negentropy \labelcref{eq:negentropy}.

While the utility of negentropy is apparent throughout the engineering disciplines, it is helpful to see why negentropy is useful from a purely information-theoretic point of view. To do so, we model an information-theoretic analogue of a heat engine. Our ``information engine'' operates in an abstract Universe of coarse-grained subsystems, with no concept of reservoirs, energy, or heat. Setting $\pi:=\sharp$ and identifying the state space of each subsystem with the set of binary strings of a fixed length, the negentropy \labelcref{eq:negentropy} of any given state $x$ reduces to approximately $|x| - K(x)$, the compressibility of $x$. Thus, compressible strings are the resource which should power the engine.

Let the engine have an internal memory system with state space $\bits^m$. Using the self-delimiting encodings \labelcref{eq:prefixfree} and $\pf x := \pf{|x|}x$, any string $x\in\bits^*$ that satisfies $|\pf x|\le m$ can be encoded in memory as the concatenation of $\pf{|x|}$, $x$, and a padding of $m-|\pf x|$ zeros, which we denote by
\[\enc(x) := \pf x\mathtt 0^{m-|\pf x|}
= \pf{|x|}x\mathtt 0^{m-|\pf x|}\in\bits^m.\]
The self-delimiting prefix $\pf{|x|}$ makes $\enc(x)$ uniquely decodable into its three parts. 

The engine uses a fixed \textbf{lossless compression algorithm}: a computable injective function $f:\bits^*\rightarrow\bits^*$, whose worst-case blowup
\[c := \max_{x\in\bits^*}\left\{|\pf{f(x)}|-|\pf{x}|\right\}\]
is much less than $m$. Since $\enc$ and $f$ are injective, the mapping
\begin{equation}
\label{eq:enc}
\enc(x) \mapsto \enc(f(x)),
\end{equation}
defined on the range of $\enc$, is also injective.

If $f$ is not simply describable, we can take it to be programmed onto a read-only section of memory, acting as a control in the simply describable joint mapping $g: (f,\,x)\mapsto (f,\,f(x))$. Thus, no generality is lost in assuming $K(f)\eqplus 0$, which implies
\[K(x)\eqplus K(f(x))\lplus |\pf{f(x)}|.\]
Therefore, $K(x)$ sets an optimistic bound on how well the mapping \labelcref{eq:enc} compresses the data in the memory. When the compression succeeds, the zero padding lengthens.

We are now ready to describe the engine's operation. Compressible strings are its fuel, to be collected from the environment, while incompressible strings are waste, to be expelled into the environment. The engine cycles between three modes:
\begin{enumerate}
    \item Consume (``burn'') zeros to perform some task, producing waste.
    \item Expel waste, and gather (``eat'') fresh fuel in its place.
    \item Refine (``digest'') fuel, producing zeros and a waste byproduct.
\end{enumerate}

The corresponding transitions to the memory state are summarized in a diagram:
\[\enc(x)
\xrightarrow{\mathrm{burn}}\enc(xy)
\xrightarrow{\mathrm{eat}}\enc(z)
\xrightarrow{\mathrm{digest}}\enc(f(z))
\]

Here, $x$ is a small string, perhaps compressed from the previous cycle. Hence, $\enc(x)$ has a large zero padding; we will see shortly how zeros are used to perform useful tasks. These tasks replace a portion of the padding with some other string $y$. If the engine expects $x$ and $y$ to be incompressible, it treats them as waste. The second stage identifies a promising location in the environment, where compressible strings might be found. With one reversible swap, the engine expels $xy$, and gathers the (hopefully) compressible string $z$ in its place, with $|z| = |x| + |y|$.  Finally, the third stage refines the fuel $z$ by compressing it, yielding additional zeros alongside the byproduct $f(z)$, which takes the role of $x$ when the cycle resets.

The zeros have many uses. One is that they pay for the processing of bad fuel: if the string $z$ turns out not to be compressible after all, then $f(z)$ may actually be longer than $z$, overwriting up to $c$ of the zeros. If this happens so often as to fully deplete the supply of zeros, the engine's behavior becomes ill-defined; in that event, we consider it to have ``starved to death''.

Otherwise, zero padding serves as a source of ancilla bits, fueling irreversible (many-to-one) operations by embedding them as reversible (one-to-one) operations \citep{sagawa2014thermodynamic}. Irreversible operations include data overwrites, error-correction, healing, and repair: each of these maps a larger number of ``bad'' states to a smaller number of ``good'' states. We saw an example of this in \Cref{sec:maxwell}, where the memory of Maxwell's demon is the ancilla that enables a transition $(0,\,x)\mapsto (m(x),\,y)$, even when the second law forbids directly mapping $x\mapsto y$. \citet{bennett1982thermodynamics} offers another example based on adiabatic demagnetization, consuming zeros to turn heat into work.

A living organism can use an information engine to support its growth and reproduction. A direct implementation of these operations would be irreversible, because they overwrite parts of the environment with copies of the organism's data \citep{devine2016understanding}. To get a reversible implementation, the engine can absorb the data that would be overwritten, into its zero padding.

\Cref{tab:creature} illustrates such an organism's operation of an information engine, for one burn-eat-digest cycle. First, it burns some of the zero padding to perform a useful function: in this case, swapping zeros onto a desired target location in the environment. Now that the target location is cleared, the organism can reversibly copy any data, such as a genome, onto it. At this point, the zero padding is almost used up. In order for the engine to recharge, it swaps in the compressible string $\mathtt{YummyAlphabetSoup}$ from the environment. Compressing this string restores the padding to a more useful length, maintaining a kind of internal homeostasis.

\begin{table}[H] 
\caption{One cycle of an information engine's operation. Its capacity is 120 bits, written as 20 Base64 characters. Incompressible strings are represented by randomly generated characters. The length prefix and the original copy of the genome are not shown. Every action is reversible.\label{tab:creature}}
\begin{tabular}{|c|c|c|}
\toprule
\textbf{Action} & \textbf{Information engine} & \textbf{Environment segment}\\
\midrule
Begin & $\mathtt{PG18Q000000000000000}$ & $\mathtt{1ksajddSG45VYummyAlphabetSoup}$
\\Burn/clear & $\mathtt{PG18Q1ksajddSG45V000}$ & $\mathtt{000000000000YummyAlphabetSoup}$
\\Reproduce & $\mathtt{PG18Q1ksajddSG45V000}$ & $\mathtt{CopyOfGenomeYummyAlphabetSoup}$
\\Eat & $\mathtt{YummyAlphabetSoup000}$ & $\mathtt{CopyOfGenomePG18Q1ksajddSG45V}$
\\Digest & $\mathtt{WAiKV000000000000000}$ & $\mathtt{CopyOfGenomePG18Q1ksajddSG45V}$
\\\bottomrule
\end{tabular}
\end{table}

We leave comparisons with real-world heat and information engines, such as those studied by \citet{leighton2024information}, to future work. A fuller analogy might assign energy values to memory states, similar to the combinatorial reservoirs of \citet{baumeler2022thermodynamics} and \citet{ebtekar2021information}.

\section{Discussion}
\label{sec:conclusion}

In order to develop ensemble-free definitions of thermodynamic quantities, we assembled ideas from stochastic thermodynamics, dynamical systems, and algorithmic information theory. The assumption of a Markovian coarse-graining reduces physical systems to time-homogeneous discrete-state Markov processes. In this setting, stochastic thermodynamics defines the stochastic and Gibbs-Shannon entropies in terms of probabilistic ensembles of physical states \citep{shiraishi2023introduction}. In many instances, these are good practical approximations of the algorithmic entropy.

To deal with cases where a suitable ensemble description is not available, we propose that thermodynamics be based on the algorithmic entropy of \emph{individual} states. \citet{levin1984randomness}'s randomness conservation law then leads to a nonequilibrium generalization of the second law of thermodynamics (\Cref{cor:secondlaw}). To ensure the accuracy of our conclusions, we carefully accounted for some ways that information can ``leak'', such as the elapsed time (which allows for Poincar\'e recurrence), and algorithmically complex dynamics (implemented by an exogenous control).

In terms of applications, we found that the algorithmic second law streamlines the analysis of Maxwell's demon, paving the way for thermodynamic analyses of all systems lacking a natural ensemble description. We followed up with an AIT perspective on Landauer and EP costs, which we hope will encourage more research in energy-efficient computing. In particular, \labelcref{eq:landauersimple} equates the long-term net heat flow with algorithmic entropy production.

Information-theoretic perspectives on thermodynamics are gaining traction in the physics of computing \citep{wolpert2024stochastic}, biology \citep{leighton2024information}, and microscopic devices more generally \citep{parrondo2023information}. In light of the growing focus on fluctuation theorems \citep{seifert2012stochastic}, we hope to find more applications for the algorithmic fluctuation inequalities derived from \Cref{thm:detailed}. These include the Zurek-Kolchinsky inequality \labelcref{eq:flucKe}, the Jarzynski inequality \labelcref{eq:jarzynski}, and the Landauer inequality \labelcref{eq:landauer}. In addition, future work might derive algorithmic versions of stochastic thermodynamics results not studied in this article, such as uncertainty relations and speed limit theorems \citep{vo2020unified}.

While this article focuses on the second law of thermodynamics, Markov processes are known to satisfy additional information-theoretic laws. We briefly made use of the information non-increase law (\Cref{thm:lawI}), which likewise has a probabilistic version \citep[\S2.8]{thomas2006elements}. An interesting consequence of this law is that any mutual information between systems in the present is traceable to a common cause in the past. This is a time-reversal asymmetry, perhaps even as fundamental as the second law; it may help us understand the perceptual, psychological, epistemic, and causal aspects of the so-called \emph{arrow of time} \citep{bennett1994complexity,ebtekar2021information,wolpert2024memory,ebtekar2024modeling}.

Causality here is meant not in the time-symmetric sense commonly associated with Einstein's relativity, but in the asymmetric sense of \citet{reichenbach1956direction}, \citet{lewis1973causation,lewis1979counterfactual}, and \citet{bell1975theory,bell1977free}, later refined by \citet{pearl2009causality}. It generalizes the Markov property to nonlinear causal topologies. In the language of physics, the causal Markov property constrains which spacetime regions $X$ and $Y$ may be statistically correlated, conditional on a third region $Z$. \citet{janzing2010causal} present an algorithmic causal Markov property, while \citet{lorenz2022quantum} and the references therein propose quantum versions. Causal modeling describes interactions between open systems. \citet{ito2013information} apply it to information thermodynamics; future work might extend this using AIT.

There is yet another general law to consider. In an effort to capture the complexity of intricate structures found in living organisms, \citet{bennett1988logical,bennett1994complexity} defines the \emph{logical depth} of $x$, at significance level $s$, to be the minimum runtime among programs, of length up to $K(x)+s$, that output $x$. He proves that the logical depth, if it increases, can only do so slowly. Thus, logically deep objects, such as genomes, are only created by gradual processes over a long span of time. Unlike entropy, which is maximized in the late Universe, we expect that logical depth is maximized at \emph{intermediate} times: late enough for its gradual accumulation, but not so late as to be destroyed by heat death \citep{antunes2006computational,aaronson2014quantifying,jeffery2020transitions}. Note that both mutual information and logical depth describe ways in which the negentropy of a system becomes difficult to extract, demanding that separated systems be reunited in the former case, and that a long computation be rewound in the latter.

In future work, it would be interesting to study the interactions between entropy non-decrease, mutual information non-increase, logical depth slow-increase, and any related laws that are as yet undiscovered. Together, they seem to characterize the arrow of time, mediating the role of information in physics, computation, and intelligent life \citep{bennett1994complexity}. In light of the known connections between data compression, inductive learning, and intelligence \citep{hutter2004universal,wallace2005statistical,hutter2024introduction}, it might be interesting to study intelligent agent behavior from the perspective of optimizing information engines along the lines of \Cref{sec:engine}.

Finally, extending algorithmic thermodynamics to incorporate quantum information remains a wide open problem. As a promising start, several quantum analogues of the description complexity have been proposed, each with different properties \citep{berthiaume2001quantum,gacs2001quantum,vitanyi2001quantum,mora2007quantum}. Just as chaos makes classical systems probabilistic (see \Cref{sec:apxmarkov}), decoherence makes quantum systems behave like mixed channels \citep[\S6.2]{sagawa2022entropy} \citep{peres1984stability,zurek1998decoherence,schlosshauer2007decoherence,esposito2010entropy}, which might help explain their irreversibility.

\begin{acknowledgements}
We are thankful for Xianda Sun's review of an early draft. In addition, the first author benefited from discussions with Jason Li, Gordon Walter Semenoff, Alexander Shen, Laurent Bienvenu, Sven Bachmann, William Unruh, Danica Sutherland, William Evans, and David Wakeham, as well as Matthew Leighton and other participants at David Wolpert and G\"ulce Karde\c s' workshop on Stochastic Thermodynamics and Computer Science Theory at the Santa Fe Institute.
\end{acknowledgements}

\appendix

\section{A Markovian coarse-graining}
\label{sec:apxmarkov}

No discussion of the second law would be complete without addressing the fundamental modeling assumptions responsible for the asymmetry between past and future. To simplify matters, consider the doubly stochastic case, corresponding to a phase space partitioned into cells of equal Liouville measure. The time reversal of a \emph{deterministic} doubly stochastic process is again doubly stochastic; by \Cref{cor:janzing}, its entropy can neither increase nor decrease at an appreciable rate.

Randomness breaks the symmetry: given a time-homogeneous doubly stochastic process, its time reversal need not be time-homogeneous nor double stochastic \citep{cover1994processes}. This matches our real-life macroscopic experience, where forward evolutions follow localized statistical laws, but backward evolutions do not. For example, a glass vase in free fall will shatter at a predictable time; and while the final arrangement of its pieces is chaotic and hard to predict, we can expect it to follow a well-defined statistical distribution. Moreover, our statistical prediction would not depend on any prior or concurrent happenings, e.g., at the neighbor's house.

In contrast, consider the time-reversed view, where we see a broken vase and want to retrodict its time of impact. It is hard to make even a meaningful statistical prediction. Our best attempt would be based on principles beyond the localized physics: for example, we might take into account a conversation at the neighbor's house, telling of the accident. In the reverse dynamics, distant shards begin to converge simultaneously, in apparent violation of locality.

Experience suggests that time-homogeneous Markov processes, despite their asymmetry, are good models of real macroscopic systems. Meanwhile, the fundamental microscopic laws of nature are widely believed to be deterministic and CPT symmetric \citep{lehnert2016cpt}. How, then, can nature's coarse-grained evolution violate this symmetry?

To demonstrate the plausibility of such an emergent asymmetry, we construct a system for which it occurs. \citet{gaspard1992diffusion} defines the \textbf{multibaker map}: a deterministic time-reversible dynamical system that, when suitably coarse-grained, emulates a random walk. \citet{altaner2012microscopic} generalize the multibaker map to emulate arbitrary Markov chains. To convey their idea in an easier fashion, we now present multibaker maps at an intermediate level of generality.

Fix an integer $m>1$. We augment the coarse-grained state space $\mathcal X$ with a bi-infinite sequence of $\ints_m$-valued microvariables, so that the total fine-grained state space is $\mathcal X\times(\ints_m)^\ints$. Every individual fine-grained state can be written in the form
\[(x,\; (\ldots,\, r_{-2},\,r_{-1},\,r_0,\,r_1,\,r_2,\,\ldots)),\]
where $x\in\mathcal X$ is the coarse-grained part, and the $r_i\in\ints_m$ collect the remaining fine-grained information. Alternatively, we can rearrange the variables and punctuation into
\[(x.r_{-1}r_{-2}r_{-3}\ldots,\; 0.r_0r_1r_2\ldots).\]

The ``0.'' here is purely symbolic. If we were to identify $\mathcal X$ with $\ints$, the latter notation is suggestive of the base $m$ representation of a point in the two-dimensional ``phase space'' $\reals \times [0,1]$. There is an extensive literature that studies symbolic representations as proxies for continuous chaotic dynamical systems; for theory and examples, see \citet{lind2021introduction}.

At each discrete time step, the system evolves by a deterministic and reversible two-stage transformation. The first stage shifts all of the $r_i$ by one position; we think of it as emulating microscopic chaos. The second stage applies a fixed bijection of $\mathcal X\times\ints_m$ to the pair $(x,\,r_0)$; we think of it as emulating the coarse-grained physics. In summary:
\begin{align*}
&(x.r_{-1}r_{-2}r_{-3}\ldots,\; 0.r_0r_1r_2\ldots)
\\\xrightarrow{\mathrm{shift}}\;
&(x.r_0r_{-1}r_{-2}\ldots,\; 0.r_1r_2r_3\ldots)
\\\xrightarrow{\mathrm{transform}}\;
&(x'.r_0'r_{-1}r_{-2}\ldots,\; 0.r_1r_2r_3\ldots).
\end{align*}

The system's only source of randomness is its initial condition. At the start time $t=0$, we allow $x$ to have any chosen distribution, but require the $r_i$ to be uniformly distributed, with all of the variables being independent. We can think of $r_i\in\ints_m$ as an $m$-sided die used to emulate a stochastic transition of $x$ at the time $t=i$. In the coarse-grained view, where we marginalize all of the $r_i$, it is easy to verify that $x$'s trajectory is a time-homogeneous doubly stochastic Markov chain, whose transition matrix entries are all multiples of $1/m$. It follows from \Cref{cor:secondlaw} that $S_\sharp(x) = K(x)$ is non-decreasing up to minor fluctuations.

In fact, our multibaker map can emulate \emph{all} such Markov chains, by a suitable choice of the bijection $T:(x,\,r_0)\mapsto (x',\,r_0')$. Indeed, recall that a Markov chain's distribution is uniquely determined by its initial condition and transition matrix. Since we already allow the initial distribution of $x$ to be arbitrary, it remains only to emulate the doubly stochastic matrix $P$. Since its entries are multiples of $1/m$, we only need to assign each pair $x,y\in\mathcal X$ to each other with multiplicity $m\cdot P(y,\,x)$. One way to accomplish this is to fix any total order $<$ on $\mathcal X$, and let
\begin{align*}
T\left(x,\; i+m\sum_{z<y} P(z,\,x)\right)
:= \left(y,\; i+m\sum_{z<x} P(y,\,z)\right)
\quad \forall x,y\in\mathcal X,\;i\in\ints_{m\cdot P(y,\,x)}.
\end{align*}

Thus, quite a diverse class of Markov chains arise as the coarse-grained part $x$ of some multibaker map. In particular, this construction realizes every example of a doubly stochastic Markov chain in this article as a deterministic time-reversible map, by appending the microvariables $r_i$ to its state.

The full construction by \citet{altaner2012microscopic} relaxes the requirement that $P$ be doubly stochastic, or that its entries have a common denominator $m$. The unpublished manuscript by \citet{ebtekar2021information} does the same in a different manner, and relaxes the requirement that the $r_i$ be uniform and independent: provided that the fine-grained state starts with a continuous distribution, it is shown that the dynamics eventually stabilize to become Markovian, time-homogeneous, and doubly stochastic. Thus, any sufficiently smooth initial distribution may serve as \citet{albert2001time}'s \emph{Past Hypothesis}. Ebtekar and Hutter \citep{ebtekar2021information,ebtekar2024modeling} provide further extensions to model causal interactions.

While these conclusions are only proven for variants of the multibaker maps, they are highly suggestive of techniques that we might try extending to realistic systems. \citet[\S4.8]{gaspard2022statistical} suggests that we should seek a short-term ergodic property of the state's microscopic part, occuring on a much faster time scale than macroscopic ergodicity. Given a suitable coarse-graining, the goal would be to prove fast convergence to Markovian behavior, long before the slower but better-understood convergence to maximum entropy. In this manner, we hope to establish the second law of thermodynamics as a mathematically rigorous property of real, CPT-symmetric systems.

\section{Conservation of randomness}
\label{sec:apxrandom}

To state the needed mathematical results in full generality, we allow the reference measure $\pi$ to be non-stationary. Denoting the probability of each state transition $x\rightarrow y$ by $P(y,\,x)$, the algorithmic entropy production is defined by
\[S_{P\pi}(y\mid\widetilde P) - S_\pi(x\mid\widetilde P),\]
where $P\pi$, $\widetilde P$ and $S_\pi$ are given by \labelcref{eq:dualpi,eq:dualP,eq:entropy}, respectively.

Before going into formal proofs, we sketch some intuition for why we expect the entropy production to be positive. For notational convenience, let $x^*:= (x,\,K(x))$ and $x^*_z:= (x,\,K(x\mid z))$. Using \labelcref{eq:negentropy}, define the \textbf{conditional randomness deficiency}
\begin{align*}
d_P(y\mid x)
&:= J_{P(\cdot,\,x)}(y\mid x^*)
= \log\frac{1}{P(y,\,x)} - K(y\mid x^*).
\end{align*}

Omitting the conditioning on $\widetilde P$ for brevity, the algorithmic entropy production during a state transition $x\rightarrow y$ can be expressed as
\begin{align*}
S_{P\pi}(y) - S_\pi(x) 
&= K(y) - K(x) + \log\frac{P\pi(y)}{\pi(x)}
\\&= K(y) - K(x) + \log\frac{P(y,\,x)}{\widetilde P(x,\,y)}
\\&\eqplus K(y\mid x^*) - K(x\mid y^*) + \log\frac{P(y,\,x)}{\widetilde P(x,\,y)}
\\&= d_{\widetilde P}(x\mid y) - d_P(y\mid x).
\end{align*}

From the last line's symmetry, one might guess that the entropy production is equally likely to be positive or negative. However, note that in general,
\[\Pr(y\mid x)=P(y,\,x),\text{ whereas }
\Pr(x\mid y)\ne\widetilde P(x,\,y).\]

While randomness deficiencies typically satisfy $d_P(y\mid x)\approx 0$, under the mismatched backward probabilities we may have $d_{\widetilde P}(x\mid y)\gg 0$. Thus, the algorithmic entropy production measures the extent to which $x$ is an ``atypical'' predecessor of $y$, when viewed as a sample from $\widetilde P(\cdot,\,y)$. For example, a low-entropy initial state $x$ would be highly atypical for the doubly stochastic matrix $P(y,\,x)=\widetilde P(x,\,y)=1/|\mathcal X|$, so the expected entropy production is high in this case.

Our formal proofs are based on the following lemma. Like the conditional randomness deficiency, $f(y,\,x)$ here can be interpreted as a conditional test of $P$-randomness for $y$, given some data $g(x)$. When $f(y,\,x)$ is large, a statistician would reject the claim that $y$ was sampled from the distribution $P(\cdot,\,x)$ \citep[\S4.3.5]{li2019introduction} \citep{ramdas2023game,ramdas2024hypothesis}. Note that neither $f$ nor $g$ are required to be computable.

In the physical interpretation, we will see that $-f(y,\,x)$ generalizes the role of entropy production during a state transition $x\rightarrow y$. \Cref{lem:detail2integral} says that the \textbf{detailed fluctuation inequality} \labelcref{eq:detailed} implies the \textbf{integral fluctuation inequalities} \labelcref{eq:integral}, the \textbf{mean bounds} \labelcref{eq:jensen}, and the \textbf{tail bound} \labelcref{eq:tail}. Thus, for any $f$, proving \labelcref{eq:detailed} is sufficient to conclude the rest.

\begin{lemma}
\label{lem:detail2integral}
Let $X,Y$ be $\mathcal X,\mathcal Y$-valued random variables, and write $P(y,\,x) := \Pr(Y=y \mid X=x)$. Let $f:\mathcal Y\times\mathcal X\rightarrow\reals$ and $g:\mathcal X\rightarrow\bits^*$ be any functions satisfying
\begin{equation}
\label{eq:detailed}
\forall x\in\mathcal X,\, y\in\mathcal Y, \quad
f(y,\,x) \le \log\frac{1}{P(y,\,x)} - K(y\mid g(x)).
\end{equation}
Then,
\begin{align}
\expect{2^{f(Y,\,X)}\mid X}
&< 1,
\qquad\expect{2^{f(Y,\,X)}}
< 1,\label{eq:integral}
\\\expect{f(Y,\,X)\mid X} &< 0,
\qquad\expect{f(Y,\,X)} < 0,\label{eq:jensen}
\end{align}
and for $\delta > 0$, with probability greater than $1-\delta$,
\begin{equation}
\label{eq:tail}
f(Y,\,X) < \log\frac 1\delta.
\end{equation}
If instead \labelcref{eq:detailed} holds with $\lplus$, then so do \labelcref{eq:jensen} and \labelcref{eq:tail}, and \labelcref{eq:integral} then holds with $\lmul$.
\end{lemma}

\begin{proof}
Rearranging \labelcref{eq:detailed},
\[\log P(y,\,x) + f(y,\,x) \le -K(y\mid g(x)).\]

To get \labelcref{eq:integral}, apply Kraft's inequality \labelcref{eq:kraft}:
\[\expect{2^{f(Y,\,X)}\mid X}
:= \sum_{y\in\mathcal Y} P(y,\,X) \cdot 2^{f(y,\,X)}
\le \sum_{y\in\mathcal Y}  2^{-K(y\mid g(X))}
< 1.\]

To get \labelcref{eq:jensen}, apply Jensen's inequality:
\[\expect{f(Y,\,X)\mid X}
\le \log\expect{2^{f(Y,\,X)}\mid X}
< 0.\]

The unconditional versions of \labelcref{eq:integral,eq:jensen} follow from the law of total expectation.  Finally, let $\delta>0$. Markov's inequality on \labelcref{eq:integral} implies that, with probability greater than $1-\delta$,
\[2^{f(Y,\,X)} < \frac 1\delta.\]

Taking logarithms now yields \labelcref{eq:tail}. The case with $\lplus$ follows similarly.
\end{proof}

As a first application of \Cref{lem:detail2integral}, consider the case where $X$ is constant, $Y$ is distributed in proportion to some system's stationary distribution $\pi$, and $g(\cdot) := \widetilde P$ (i.e., $g$ is a constant function that outputs a program computing $\widetilde P$). Then, the right-hand side of \labelcref{eq:detailed} reduces to the negentropy \labelcref{eq:negentropy}, and the conclusion \labelcref{eq:tail} agrees with \labelcref{eq:negentropysaturates}.

Next, we turn to general nonequilibrium dynamics. We present detailed fluctuation inequalities for the change in each of the quantities $K$, $\pi$, $S_\pi$, and $I$.

\begin{theorem}[Detailed fluctuations]
\label{thm:detailed}
Let $\pi:\mathcal X\rightarrow\reals^+\setminus\{0\}$ be a measure. When they appear as side information, suppose $P:\mathcal Y\times\mathcal X\rightarrow\reals^+$ and $\widetilde P$ from \labelcref{eq:dualP} are computable. Then, for all $x\in\mathcal X$, $y\in\mathcal Y$, and $z\in\bits^*$,
\begin{align}
K(y\mid P) - K(x\mid P) &\lplus \log\frac{1}{P(y,\,x)} - K(x\mid y^*_P,P),\label{eq:indivK}
\\\log\pi(x) - \log P\pi(y) &\lplus \log\frac{1}{P(y,\,x)} - K(x\mid y,\widetilde P),\label{eq:indivpi}
\\S_\pi(x\mid\widetilde P) - S_{P\pi}(y\mid\widetilde P) &\lplus \log\frac{1}{P(y,\,x)} - K(y\mid x^*_{\widetilde P},\widetilde P),\label{eq:indivS}
\\I(y\mi z\mid P) - I(x\mi z\mid P) &\lplus \log\frac{1}{P(y,\,x)} - K(y\mid (x,z)^*_P,P).\label{eq:indivI}
\end{align}
\end{theorem}

\begin{proof}
For each $x$, $P(\cdot,\,x)$ is a probability measure computable by a constant-sized program along with $(x,P)$; similarly, $\widetilde P(\cdot,\,y)$ can be computed using $(y,\widetilde P)$. Hence, \labelcref{eq:loguniversal} implies
\begin{align*}
K(y\mid x,P) \lplus \log\frac{1}{P(y,\,x)},
\quad K(x\mid y,\widetilde P) \lplus \log\frac{1}{\widetilde P(x,\,y)}.
\end{align*}

Now, we verify the inequalities one at a time. \labelcref{eq:indivK} follows from
\begin{align*}
K(x\mid y^*_P,P)
&\eqplus K(x,y\mid P) - K(y\mid P)
\\&\lplus K(y\mid x,P) + K(x\mid P) - K(y\mid P)
\\&\lplus \log\frac{1}{P(y,\,x)} + K(x\mid P) - K(y\mid P).
\end{align*}

Similarly, \labelcref{eq:indivpi} follows from
\begin{align*}
K(x\mid y,\widetilde P)
&\lplus \log\frac{1}{\widetilde P(x,\,y)}
\\&= \log\frac{1}{P(y,\,x)} + \log P\pi(y) - \log\pi(x).
\end{align*}

The proof of \labelcref{eq:indivS} combines the steps of the previous derivations:
\begin{align*}
K(y\mid x^*_{\widetilde P},\widetilde P)
&\eqplus K(x,y\mid\widetilde P) - K(x\mid\widetilde P)
\\&\lplus K(x\mid y,\widetilde P) + K(y\mid\widetilde P) - K(x\mid\widetilde P)
\\&\lplus \log\frac{1}{\widetilde P(x,\,y)} + K(y\mid\widetilde P) - K(x\mid\widetilde P)
\\&= \log\frac{1}{P(y,\,x)} + S_{P\pi}(y\mid\widetilde P) - S_\pi(x\mid\widetilde P).
\end{align*}

\labelcref{eq:indivI} was first shown by \citet{gacs2001algorithmic}, but we present a simpler proof based on \citet{gacs2021lecture}. Applying the algorithmic data processing identity \labelcref{eq:dataproc} twice,
\begin{align*}
K(y\mid (x,z)^*_P,P)
&\eqplus K(y\mid x^*_P,P) - I(y\mi z\mid x^*_P,P)
\\&\eqplus K(y\mid x^*_P,P) + I(x\mi z\mid P) - I((x,y)\mi z\mid P)
\\&\lplus \log\frac{1}{P(y,\,x)} + I(x\mi z\mid P) - I(y\mi z\mid P).
\end{align*}
\end{proof}

\Cref{thm:detailed} says that if a transition $x\rightarrow y$ occurs with substantial probability $P(y,\,x)$, then it cannot substantially \emph{decrease} the Liouville measure $\pi$ or the algorithmic entropy $S_\pi$, nor can it substantially \emph{increase} the description complexity $K$ or the algorithmic mutual information $I$ (with respect to any fixed object $z$). This does \emph{not} imply that the quantities trend monotonically, since a large number of low-probability transitions may still sum to a high probability.

For example, consider a Markov chain that alternates between a ``hub'' state, and a uniformly random selection among a large number $m$ of other states. Formally, let $\mathcal X:=\ints_{m+1}$; for $x,y=1,\,\ldots,m$, let
\[\pi(0):=m,
\;\pi(x):=1,
\;P(0,\,0):=0,
\;P(y,\,0):=1/m,
\;P(0,\,x):=1,
\;P(y,\,x):=0.\]
Then, $P$ is $\pi$-stochastic. The hub state has $K(0\mid\widetilde P)\eqplus 0$, while most of the other states have $K(x\mid\widetilde P)\eqplus\log m$. Therefore, both $\pi$ and $K$ are highly non-monotonic, taking turns alternating between a much lower and a much higher value.

On the other hand, it is easy to check that the hub state, as well as most of the other states, have approximately $\log m$ entropy. There are a few states with low entropy: for example, $S_\pi(1\mid\widetilde P)\eqplus\log\pi(1)=0$. If we start from such a state, the entropy will immediately increase to about $\log m$, and we will seldom return to these rare low-entropy states.

In general, $S_\pi$ and $I$ may fluctuate a bit but, unlike $\pi$ and $K$, they trend monotonically. The non-decrease law for $S_\pi$ is known as \textbf{randomness conservation}, while the non-increase law for $I$ is called \textbf{information non-increase}. Both were first shown by \citet{levin1984randomness}. Our statement and proof take after the more pedagogical exposition of \citet{gacs2021lecture}, though we make additional changes. One is that we have extracted \Cref{lem:detail2integral} as a general tool, to derive these integral fluctuation inequalities from their detailed counterparts. Another is that we condition on the dual matrix $\widetilde P$, to eliminate some error terms from the older results. To eliminate $\widetilde P$ altogether, see \Cref{sec:refine}.

Note that in the main body of this article, we always have $\mathcal X=\mathcal Y$ and $P\pi=\pi$. Using \labelcref{eq:dualP}, it follows that $\widetilde P$ is computable if both $\pi$ and $P$ are.

\begin{theorem}[Randomness conservation]
\label{thm:lawS}
Let $\pi:\mathcal X\rightarrow\reals^+\setminus\{0\}$ be a measure and $X,Y$ be $\mathcal X,\mathcal Y$-valued random variables. Suppose $\widetilde P$, defined in terms of $P(y,\,x) := \Pr(Y=y \mid X=x)$ by \labelcref{eq:dualP}, is computable. Then,
\begin{equation*}
\expect{2^{S_\pi(X\mid\widetilde P) - S_{P\pi}(Y\mid\widetilde P)}} \lmul 1.
\end{equation*}
Therefore, for $\delta > 0$, with probability greater than $1-\delta$,
\begin{equation*}
S_\pi(X\mid\widetilde P) - S_{P\pi}(Y\mid\widetilde P) \lplus \log\frac 1\delta.
\end{equation*}
\end{theorem}

\begin{proof}
Let $f(x,\,y) := S_\pi(x\mid\widetilde P) - S_{P\pi}(y\mid\widetilde P)$ and $g(x) := (x^*_{\widetilde P},\widetilde P)$. Then, \labelcref{eq:indivS} from \Cref{thm:detailed} implies that the hypotheses of \Cref{lem:detail2integral} hold, and therefore so do its conclusions.
\end{proof}

Finally, physical applications motivate us to frame the information non-increase law in terms of two independently evolving systems.

\begin{theorem}[Information non-increase]
\label{thm:lawI}
For $i=1,2$, let $P_i:\mathcal Y_i\times\mathcal X_i\rightarrow\reals^+$ be computable stochastic matrices, and $X_i,Y_i$ be $\mathcal X_i,\mathcal Y_i$-valued random variables, such that for all $x_i\in\mathcal X_i$ and $y_i\in\mathcal Y_i$,
\begin{equation*}
\label{eq:indepmatrix}
\Pr((Y_1,Y_2)=(y_1,y_2) \mid (X_1,X_2)=(x_1,x_2)) = P_1(y_1,\,x_1)P_2(y_2,\,x_2).
\end{equation*}
Then, writing $P:=(P_1,\,P_2)$,
\begin{equation*}
\label{eq:flucI}
\expect{2^{I(Y_1\mi Y_2\mid P)-I(X_1\mi X_2\mid P)}} \lmul 1.
\end{equation*}
Therefore, for $\delta > 0$, with probability greater than $1-\delta$,
\begin{equation*}
\label{eq:lawI}
I(Y_1\mi Y_2\mid P)-I(X_1\mi X_2\mid P) \lplus  \log\frac 1\delta.
\end{equation*}
\end{theorem}

\begin{proof}
The pair $P$ can be affixed with a constant-sized instruction to compute either $P_1$ or $P_2$. Applying \labelcref{eq:indivI} from \Cref{thm:detailed} twice, first with $z=x_2$ and then with $z=y_1$:
\begin{align*}
I(y_1\mi x_2\mid P) - I(x_1\mi x_2\mid P)
&\lplus \log\frac{1}{P_1(y_1,\,x_1)} - K(y_1\mid (x_1,x_2)^*_P,P),
\\I(y_1\mi y_2\mid P) - I(y_1\mi x_2\mid P)
&\lplus \log\frac{1}{P_2(y_2,\,x_2)} - K(y_2\mid (y_1,x_2)^*_P,P)
\\&\lplus \log\frac{1}{P_2(y_2,\,x_2)} - K(y_2\mid (y_1,x_1,x_2)^*_P,P).
\end{align*}

Summing these inequalities yields
\begin{align*}
I(y_1\mi y_2\mid P) - I(x_1\mi x_2\mid P)
&\lplus \log\frac{1}{P_1(y_1,\,x_1)P_2(y_2,\,x_2)} - K(y_1,y_2\mid (x_1,x_2)^*_P,P).
\end{align*}

Let $f((x_1,x_2),(y_1,y_2)) := I(y_1\mi y_2\mid P) - I(x_1\mi x_2\mid P)$ and $g((x_1,x_2)) := ((x_1,x_2)^*_P,P)$. The desired result now follows from \Cref{lem:detail2integral}.
\end{proof}

\bibliographystyle{apsrev4-1}
\bibliography{main}

\end{document}